\documentclass[12pt,a4paper]{article}
\usepackage[affil-it]{authblk}
\usepackage[dvips]{epsfig}
\usepackage{amscd}
\usepackage{amssymb}
\usepackage{amsthm}
\usepackage{amsmath}
\usepackage{latexsym}
\usepackage{color}
\usepackage{setspace}
\usepackage{times}
\usepackage{pdflscape}

\newtheorem{theorem}{Theorem}
\theoremstyle{plain}

\newtheorem{example}[theorem]{Example}

\newtheorem{lemma}[theorem]{Lemma}

\newtheorem{proposition}[theorem]{Proposition}

\newtheorem{remarks}[theorem]{Remarks}

\newcommand{\R}{\ensuremath{\mathbb{R}}}
\newcommand{\E}{\ensuremath{\mathbb{E}}}

\def\e{{\mathrm{e}}}
\newcommand*{\ud}{\mathrm{d}}
\allowdisplaybreaks

\begin{document}
\title{Pricing of commodity derivatives on processes with memory}
\author{Fred Espen Benth$^{1}$}
\author{Asma Khedher$^{2}$}
\author{Mich\`ele Vanmaele$^{\;\,3}$}
\affil{\small $^{1}$Department of Mathematics, University of Oslo, P.O. Box 1053, Blindern, N-0316 Oslo, Norway\\$^{2}$Korteweg-de Vries Institute for Mathematics, P.O. Box 94248, 1090 GE Amsterdam, The Netherlands\\
$^{3}$Department of Applied Mathematics, Computer Science and Statistics, Ghent University, Krijgslaan 281 S9, B-9000 Gent, Belgium}


\maketitle
\begin{abstract}
Spot option prices, forwards and options on forwards relevant for the commodity markets are computed when the underlying process $S$ is modelled as an exponential of a process $\xi$ with memory as e.g.\ a L\'evy semi-stationary process. Moreover a risk premium $\rho$ representing storage costs, illiquidity, convenience yield or insurance costs is explicitly modelled as an Ornstein-Uhlenbeck type of dynamics with a mean level that depends on the same memory term as the commodity. Also the interest rate is assumed to be stochastic. To show the existence of an equivalent pricing measure $\mathbb{Q}$ for $S$ we relate the stochastic differential equation for $\xi$ to the generalised Langevin equation.
When the interest rate is deterministic  the process $(\xi; \rho)$ has an affine structure under the pricing measure $\mathbb{Q}$ and  an explicit expression for the option price is derived in terms of the Fourier transform of the payoff function. 
\end{abstract}


\section{Introduction}

In financial markets the arbitrage-free price of a derivative is derived by a risk neutral probability. In complete markets, the risk-neutral probability
is unique, leading to a single arbitrage-free price dynamics. Most financial markets are, however, incomplete, with commodity markets as a
typical case. For example, the power spot market is only accessible for physical players that can produce or transmit electricity, whereas the 
forwards market on power is financial. See \cite{BenthKoe,Eydel,Geman} for a discussion on pricing in energy and commodity markets, and
\cite{Bjork} for a general treatment of the arbitrage pricing theory in financial markets. Power serves as the extreme example of an incomplete market, as the spot is considered not financially tradeable in addition to a price dynamics with highly non-Gaussian features such as price spikes. 

There is no unique risk-neutral probability in an incomplete market. In this paper, we focus on a class of probabilities $\mathbb Q$ that can be represented as a ``deviation'' from the risk-neutral one, in the sense that the price dynamics of the underlying asset will have a mean rate
of return which can be represented as the sum of a risk-free interest rate and an additional yield under this probability 
$\mathbb Q$. The probability $\mathbb Q$ will not be risk-neutral, but only equivalent to the market probability $\mathbb P$. It is referred
to as a pricing measure. With this class of probabilities, we model the risk premium by the additional yield. In commodity markets, this yield can be 
interpreted as storage costs, transportation, insurance, convenience yield and other illiquidity costs. Thus, the market price of risk is viewed as the 
compensation for financial risk and illiquidity risk.

Mean reversion and stationarity play an important role in the price dynamics of commodity prices (see \cite{Eydel,Geman}). We consider a 
spot price model where the logarithmic price dynamics follows a generalised Langevin equation. In our framework, we allow for
dependency on the past in the current price, as well as jumps. The dependency on the past in the dynamics comes in as
a memory term in the drift of the Langevin equation, being modelled as a weighted average of the historical logarithmic prices. 
Our model includes the class of L\'evy semi-stationary processes and 
continuous-time autoregressive moving average processes, popular modelling tools for power, gas and oil prices 
(see \cite{BBV3,BKMV,PP}) and weather variables like temperature and wind (see \cite{Benth}), as well
as volatility and turbulence (see \cite{BS}). We prove existence and uniqueness of a solution of our proposed general Langevin equation model.   

Our main result is that the class of pricing measures that we propose are indeed probabilities. This entails in proving that the density process in the Girsanov theorem is a true martingale. We appeal to the criteria in the extended Bene\v{s} method, developed in \cite{Klebaner}.    
In our analysis, we allow for stochastic interest rates and a stochastic dynamics for the yield in the risk premium. Both processes are modelled by jump-diffusions of Ornstein-Uhlenbeck type, with an explicit dependency on the memory part of the Langevin dynamics in the price 
dynamics. 

We perform an in-depth study of pricing of options and forwards using our pricing measure in the special cases when the logarithmic spot price
dynamics is of L\'evy semi-stationary type or a continuous-time autoregressive moving average process. Due to the affine structure, we obtain reasonably explicit expressions for call and put option prices in the former case using Fourier methods. For L\'evy semi-stationary processes, the forward price is also available explicitly as a function of the spot and the risk premium. Furthermore, we express the price dynamics for put and call options on forwards for this model class. Plain vanilla European options are typically traded on forwards in many power markets.
Further, we consider Wiener-driven continuous-time autoregressive moving average dynamics and extend the results in \cite{Benth} to introduce a class of pricing measures $\mathbb Q$. For these models, we derive the forward price dynamics under our pricing measure. 
We remark that the explicit price expressions in all cases are derived under the assumption of deterministic interest rates and hence future and forward prices coincide. 

Our analysis and results are presented as follows. Section 2 provides some motivation from commodity markets on pricing and risk-neutral probabilities, and presents the stochastic dynamics of the market that we will analyse in this paper. As background material, we also include some results on affine processes that will be needed later in the paper. The generalised Langevin equation modelling the logarithmic spot prices is analysed in Section 3, and in Section 4 we prove the validity of our proposed measure change.  Finally, Section 5 derives prices for various derivatives like options and forwards in the case of L\'evy semi-stationary processes. An extensive analysis for the case of continuous-time autoregressive moving average processes is also contained in this section.

%
\section{Set up and preliminaries}\label{sect1}


Suppose that $(\Omega,\mathcal{F},(\mathcal{F}_t )_{0\leq t\leq T},\mathbb{P})$ is a given probability space satisfying the usual conditions, see e.g.\ \cite{Protter}, and $W = (W(t))_{0\leq t\leq T}$ denotes an $(\mathcal{F}_t)$-Wiener process. Furthemore, we let 
$L = (L(t))_{0\leq t\leq T}$ be an $(\mathcal{F}_t)$-L\'evy process with characteristic triplet $(\varsigma, c^2, \ell)$.
Assume that the L\'evy measure of the process $L$ satisfies $\int_{|z|\geq 1} z^2 \,\ell(\ud z) < \infty$, e.g., $L$ is a
square-integrable L\'evy process. From the characteristic triplet of the process $L$, we know that the latter admits the following decomposition 
\begin{equation}\label{Levy}
L(t) = b t + cW (t) + \int\limits_0^t  \int\limits_\R z  \tilde{N} (\ud s, \ud z), \qquad  t\in [0,T], 
\end{equation}
where $b= \varsigma+ \int_{|z|\geq 1} z \, \ell(\ud z)$ and $\tilde{N}$ is a compensated  Poisson random measure. That is $\tilde{N}(\ud t, \ud z) = N(\ud t,\ud z) - \ell(\ud z)  \ud t$  and $N(\ud t, \ud z)$
is the Poisson random measure such that $\E[N(\ud t,\ud z)] = \ell(\ud z)  \ud t$.

Before defining our stochastic model for the commodity spot market (see Subsection~\ref{subsect:model}), we discuss 
some classical findings of forward pricing. 

\subsection{Spot and forwards in commodity markets}
Let $S$ be a stochastic process defining the spot price dynamics of a commodity given by a geometric Brownian motion under $\mathbb{P}$
\begin{equation}\label{GBM}
\ud S(t)=\mu S(t) \,\ud t+\sigma S(t) \,\ud W(t),
\end{equation}
where $\mu,\sigma>0$ are constants. 
If the spot can be liquidly traded in the commodity market, then we can perfectly hedge a short position in a forward contract by a long position in the spot financed by borrowing at the risk-free rate $r$. This hedging strategy is known as the buy-and-hold strategy and uniquely defines the forward prices (see e.g.~\cite{Eydel,Geman}). I.e., if $F(t,T)$ is the forward price at time $t\geq 0$ of a contract delivering the commodity at time
$T\geq t$, then $F(t,T)=S(t)\exp(r(T-t))$.  

When running a buy-and-hold strategy in a commodity market, the commodity must be stored. Thus the hedger will be incurred additional costs reflected in the forward price as an increased interest rate to be paid.  On the other hand, holding the commodity has a certain advantage over being long a forward contract due to the greater flexibility. The notion of convenience yield is introduced to explain this additional benefit accrued to the owner of the physical commodity (see \cite{Eydel,Geman} for more details on convenience yield). Denoting $\rho$ the yield from storage and convenience, one derives a forward price $F(t,T)=\exp((r+\rho)(T-t))$ for $0\leq t\leq T$.  
In contrast to a classical commodity market,
agriculture say, where the insurer can hedge her risk to some extent, the spot power market on the other hand is completely unhedgeable. Speculators cannot hedge in the power spot, where only physical players can take part. In this respect, a forward on power can be seen as a pure insurance instrument, where a speculator can offer insurance to a producer. To describe the added premium that the insurer charges to take on the risk one can also use an increased interest rate $r+\rho$ as above. We may refer to $\rho$ as the {\it risk premium}. 

We will introduce a pricing measure $\mathbb Q$ in terms of an explicit risk premium $\rho$ that should explain storage costs, illiquidity, and convenience yield.
Let us assume that this $\rho$ and an interest rate $r$ are constant rates.
Moreover, consider the measure change 
$$
\frac{\ud \mathbb{Q}}{{\ud \mathbb{P}}}\Big|_{\mathcal{F}_t}= \exp\Big(-\int\limits_0^t \theta\, \ud W(s) -\frac{1}{2}\int\limits_0^t \theta^2\, \ud s\Big)\,,\qquad 0\leq t\leq T\,.
$$
where $$\theta= \frac{\mu-r-\rho}{\sigma}\,. $$
Under this new probability measure $\mathbb{Q}$ it follows from Girsanov's theorem that the rate of return of $S$ equals $r+\rho$ and its 
dynamics is given by
\begin{equation}\label{eq2}
\ud S(t)= (r+\rho) S(t) \,\ud t+\sigma S(t)\, \ud W_{\mathbb{Q}}(t),\qquad 0\leq t\leq T\,,
\end{equation}
where $W_{\mathbb{Q}}$ is a Wiener process under $\mathbb{Q}$, with 
$$\ud W_{\mathbb{Q}}(t) =\ud W(t) +\frac{\mu-r-\rho}{\sigma}\, \ud t\,,\qquad 0\leq t\leq T\,.$$ 
This implies that the discounted spot process of $S$ is given under $\mathbb{Q}$ by
\[
\ud(\mbox{e}^{-rt}S(t))= \mbox{e}^{-rt}\sigma S(t) \,\ud W_{\mathbb{Q}}(t) + \mbox{e}^{-rt}S(t)\rho\, \ud t\,.
\]
Note that 
the rate of return in the $\mathbb{P}$-dynamics of the logreturns of $\tilde{S}(t) := \mbox{e}^{-rt}S(t)$ equals $\mu-r-\sigma^2/2$, for $t\in [0,T]$. 
Defining $\xi(t) :=\log S(t)$, then 
\begin{equation}\label{eq9}
\ud \xi(t) = \Big(\mu-\frac{\sigma^2}{2}\Big)  \ud t +\sigma\, \ud W(t).
\end{equation}

The forward price $F(t,T)$ contracted at $t$ with time of delivery $T\geq t$ is defined such that (see e.g. \cite{BenthKoe})
$$\E^{\mathbb{Q}}\left[\mbox{e}^{-r(T-t)} \Big(F(t,T) -S(T)\Big) \mid \mathcal{F}_t\right] =0\,. $$
When the dynamics of $S$ is given by \eqref{eq2} the latter is equivalent to 
\begin{align*}
F(t,T) &= \E^{\mathbb{Q}}\left[S(T) | \mathcal{F}_t\right] =\mbox{e}^{(r+\rho)(T-t)} S(t)\,.
\end{align*}
From the latter we see that the standard forward pricing theory is using a \textit{market price of risk} $(\mu-r-\rho)/\sigma$, where $(\mu-r)/\sigma$ is the risk-neutral change and 
$\rho/\sigma$ is added due to storage costs and convenience yield, or an insurance premium for illiquidity.
In this paper we will introduce such pricing measures for much more general models for the spot than a simple geometric Brownian motion.


\subsection{A commodity spot market model with memory and jumps}
\label{subsect:model}



We introduce our spot price dynamics $(S(t))_{0\leq t\leq T}$ in a commodity market as follows. Let 
\begin{equation}
\label{spot-model}
\xi(t):=\log S(t)\,,
\end{equation}
with $(\xi(t))_ {0\leq t\leq T}$ being a generalised Langevin equation of the form
\begin{equation}\label{eq10}
\ud\xi(t)=\Big(\int\limits_{0}^tM(t-u)\xi(u) \ud u\Big) \ud t+ \chi(t-) \,\ud L(t),
\end{equation}
where $M$ is a deterministic function and $\chi$ is a strictly positive $\mathbb{F}$-adapted c\`adl\`ag process. 
The notation $\chi(t-)$ means $\lim_{s\uparrow t}\chi(s)$, i.e., the left-limit of the process. 

Furthermore, we consider a stochastic interest rate $r$ and an explicit risk premium process  $\rho$ given by a bivariate Ornstein-Uhlenbeck (OU) type of dynamics with a mean level that depends on $\int_0^t M(t-u) \xi(u) \, \ud u$. That is 
\begin{align}
\ud r(t) &= \left[A(t) - B_2(t)r(t) \right] \ud t + B_1(t) \chi(t) c \, \ud W(t) + B_1(t) \chi(t-) \int\limits_\R z\tilde{N}(\ud t , \ud z)\,,\label{r-jumps}\\
\ud \rho(t) &= \Big[\bar{A}(t)+ \bar{B}_1(t)\int_0^t\limits M(t-u) \xi(u) \, \ud u - \bar{B}_2(t)r(t) - \bar{B}_3(t)\rho(t)  \Big] \ud t + \bar{B}_1(t) \chi(t) c \, \ud W(t)\nonumber\\
&\qquad   + \bar{B}_1(t) \chi(t-) \int\limits_\R z\tilde{N}(\ud t , \ud z)\,,  \label{rho-jumps}\\ 
r(0) &= r_0 \in \R\,, \quad \rho(0) = \rho_0 \in \R\,, \nonumber
\end{align} 
where $B_i$, $\bar{B}_i$, for $i\in \{1,2,3\}$ and $A$, $\bar{A}$ are deterministic functions uniformly bounded in $t$ by a constant.
This type of dynamics \eqref{eq10}-\eqref{rho-jumps} involving a memory is relevant for many commodity markets including power (see e.g.~\cite{BBV1} and \cite{BBV2}). The interest rate is driven by the same stochastic factors $W$ and $\tilde{N}$ as $\xi$ and $\rho$. This model is more general than a deterministic model and this particular choice is motivated by the aim to prove the existence of a measure change
(see Section~\ref{change}). In the next section we will state conditions ensuring the existence and uniqueness of a solution to the generalised Langevin equation \eqref{eq10}.    

\subsection{Affine processes}
As affine processes will play an important role in our considerations, we include in this section with preliminaries some useful results on this
class of processes.

Affine processes are continuous-time Markov processes characterised by the fact that their log-characteristic function depends in an affine way on the initial state vector of the process.
Recently affine models have gained significant attention in the finance literature mainly due to 
their analytical tractability (see for example \cite{DDW}, \cite{DPS}, and \cite{KM}). 
In the sequel we introduce an equivalent martingale measure $\mathbb{Q}$ under which the process $(\xi, r,\rho)$ introduced in \eqref{eq10}-\eqref{rho-jumps} is going to be of time-inhomogeneous affine type. Since we are interested in pricing contingent claims written on a process $\exp(\xi)$, we recall in this subsection a result showing that pricing contingent claims in time-inhomogeneous affine models can be reduced to the solution of a set of Riccati-type ordinary differential equations. 

Denote by $\cdot^{\top}$ the transpose of a given vector or a matrix. Assume there exists a unique solution to the following stochastic differential equation (SDE)
\begin{align*}
\ud X(t) &= \varpi(t,X(t)) \, \ud t + \sigma(t, X(t)) \, \ud \tilde{W}(t) + \int\limits_\R \iota(t,z) \left(\mu^X(\ud t, \ud z)-\nu(\ud t, \ud z)\right)\,,\\
X(0) &= x \in \R^d\,,
\end{align*}
where $\tilde{W}$ is a $d$-dimensional Brownian motion, for $d \geq 1$, $\mu^X$ is a random measure of the jumps of $X$, $\nu(\ud t, \ud z) = \ell_t(\ud z)\ud t$ is the compensator of the jump measure $\mu^X$ which we assume to be deterministic,
$\varpi: [0,T] \times \R^d \rightarrow \R^d$, $\sigma: [0,T] \times \R^d \rightarrow \R^{d \times d}$ is continuous and such that 
$\varrho(t,x) = \sigma(t,x) \sigma(t,x)^{\top}$ is continuous for $t \in [0,T]$ and $x\in \R^d$. Moreover, $\iota(t, z) \in \R^d$ is continuous in $t\in [0,T]$ for $z\in \R$ and such that 
$\int_\R \iota^2(t,z)\ell_t(\ud z) <\infty$, for all $t \in [0,T]$.

Moreover, assume the following affine structure for the time-dependent parameters of the SDE $X$ 
\begin{align*}
\varrho(t,x) &= \varrho(t) + \sum_{i=1}^d x_i \alpha_i(t)\,,\\
\varpi(t,x) &= \varpi(t) + \sum_{i=1}^d x_i \beta_i(t)\,,
\end{align*}
where $\varrho(t)$ and $\alpha_i(t)$ are $d\times d$ matrices and $\varpi(t)$ and $\beta_i(t)$ are $d$-vectors.  
We consider a real-valued process $R$ for which we impose the following affine structure 
$$R(t)  = c + \gamma^{\top} X(t)\,,$$
for $c \in \R$ and $\gamma \in \R^d$. 

Let $u \in \mathbb{C}^d$ and $(\phi(\cdot, T, u), \psi(\cdot, T, u)): [0,T] \rightarrow \mathbb{C}\times \mathbb{C}^d$ be $C^1$-functions satisfying the following Riccati equations
\begin{align}\label{Riccati-equations}
\partial_t \phi (t,T,u) &= -\psi(t,T,u)^{\top} \varpi(t) -\frac{1}{2} \psi(t,T,u)^{\top}\varrho(t) \psi(t,T,u)\nonumber\\
&\qquad  - \int\limits_\R[\mbox{e}^{\psi(t,T,u)^{\top} \iota(t,z)} - 1- \psi(t,T,u)^{\top} \iota(t,z)]\,\ell_t(\ud z) + c\,,\nonumber\\
\partial_t \psi_i(t, T, u) &=   - \beta_i(t)^{\top}\psi(t,T,u) -\frac{1}{2} \psi(t,T,u)^{\top}\alpha_i(t) \psi(t,T,u) + \gamma_i\,, \qquad 1\leq i \leq d, \nonumber\\
\phi(T,T,u) &= 0\,,\nonumber\\
\psi(T,T,u) &= u\,.
\end{align}
We compute in the following theorem the discounted moment generating function of $X(T)$, conditional on the information at time $t\leq T$ in terms of the solution to the Riccati equations \eqref{Riccati-equations}. For a proof we refer to \cite[Theorem 2.13]{DF}. See also \cite[Theorem 5.1]{KMK}.
\begin{theorem}\label{Affine-process-charac}
Let $u \in \mathbb{C}^d$. Suppose that \eqref{Riccati-equations} admits a unique solution $(\phi(\cdot, T, u), \psi(\cdot, T, u)): [0,T] \rightarrow \mathbb{C}\times \mathbb{C}^d$. Then 
\begin{equation*}
\E[\e^{-\int_t^T R(s) \, \ud s} \e^{u^{\top} X(T)} | \mathcal{F}_t]  = \e^{\phi(t,T,u) +\psi(t,T,u)^{\top} X(t)}, \qquad t\leq T\,.
\end{equation*}
\end{theorem}
This result allows for the use of Fourier transform techniques to compute derivative prices written on affine models, as we will return to in Section~\ref{section5}.

\section{Analysis of the generalised Langevin equation}\label{Langevin}
In this section we start from a generalised Langevin equation and its solution found by Laplace transformation to propose a solution to the corresponding SDE introduced in \eqref{eq10}. 

Consider the following generalised Langevin equation as in \cite{Fox} 
\begin{equation}\label{eq01}
\dot{\eta}(t)=\int\limits_0^tM(t-u)\eta(u) \ud u +\sigma w(t),
\end{equation}
where $M$ stands for the memory kernel, $w$ represents Gaussian fluctuations, $\sigma$ is a constant
and $\dot{\eta}$ denotes the time-derivative of $\eta$. When $M$ is the Dirac delta function the solution will be Markovian. 
Using Laplace transforms and the definition
\begin{equation}\label{LaplaceM}
\hat{M}(z)=\int\limits_0^{+\infty}\mbox{e}^{-zt}M(t)\ud t,
\end{equation}
for $z \in \mathbb{C}$, for which the integral makes sense, the solution to \eqref{eq01} can be expressed as (see \cite{Fox}),
\begin{equation}\label{eq02}
\eta(t)=H(t)\eta(0)+\sigma \int\limits_0^{t}H(t-u)w(u) \ud u,
\end{equation}
where $H$ is defined through its Laplace transform
\begin{equation}\label{eq02i}
\hat{H}(z)=\frac{H(0)}{z-\hat{M}(z)},
\end{equation}
with, for simplicity, $H(0)=1$. From \eqref{eq02i} it follows that
\begin{equation}\label{eq02ii}
\dot{H}(t)=\int\limits_0^tM(t-u)H(u) \ud u.
\end{equation}
We include here an example of a memory kernel $M$ that defines a corresponding function $H$ although it is singular at zero.
\begin{example}\label{exampleM}
Let $M(s)=s^{-\alpha}$ for $0<\alpha<1/2$. Then we should find a function $H$ satisfying \eqref{eq02ii}. This means, after integrating both sides and invoking Fubini,
\begin{equation}\label{H-int}
\begin{aligned}
H(t)-1=\int_0^t\dot{H}(s)\ud s &=\int_0^t\int_0^s(s-u)^{-\alpha}H(u)\ud u\ud s \\
&=\int_0^t\int_u^t(s-u)^{-\alpha}\ud s H(u)\ud u \\
&=\frac1{1-\alpha}\int_0^t(t-u)^{1-\alpha}H(u)\ud u.
\end{aligned}
\end{equation}
To solve the Volterra integral equation, we propose that 
\begin{equation} \label{def-H-series}
H(t)=\sum_{n=0}^{\infty}b_n(\alpha)(t^{2-\alpha})^n.
\end{equation}
Immediately, we find $b_0(\alpha)=1$. Inserting \eqref{def-H-series} into the integral equation \eqref{H-int}, we get
$$
1+\frac1{1-\alpha}\sum_{n=0}^{\infty}b_n(\alpha)\int_0^t(t-s)^{1-\alpha}s^{2n-n\alpha}\ud s
=\sum_{n=0}b_n(\alpha)(t^{2-\alpha})^n .
$$
But the integral on the left hand side is connected to the Beta-distribution:
\[
\int_0^t(t-s)^{1-\alpha}s^{\beta}\ud s=t^{2+\beta-\alpha}\int_0^1(1-u)^{1-\alpha}u^{\beta}\ud u =t^{2+\beta-\alpha}\frac{\Gamma(2-\alpha)\Gamma(1+\beta)}{\Gamma(3+\beta-\alpha)}.
\]
With $\beta=2n-n\alpha$ we find the recursive relations
$$
b_n(\alpha)=b_{n-1}(\alpha)\frac{\Gamma(2-\alpha)\Gamma(1+(n-1)(2-\alpha))}{(1-\alpha)\Gamma(1+n(2-\alpha))}.
$$
Using that $\Gamma(1+1-\alpha)=(1-\alpha)\Gamma(1-\alpha)$, we reach that
$b_0(\alpha)=1$ and
\begin{equation}
\label{def-rec-b}
b_n(\alpha)=b_{n-1}(\alpha)\frac{\Gamma\Bigl(1+(n-1)(2-\alpha)\Bigr)}{\Gamma\Bigl(1+n(2-\alpha)\Bigr)}\Gamma(1-\alpha), n=1,2,3,\ldots .
\end{equation}
Now, we may ask whether the representation \eqref{def-H-series} of $H(t)$ is a convergent series. By the ratio test we find
$$
\frac{b_n(\alpha)(t^{2-\alpha})^n}{b_{n-1}(\alpha))(t^{2-\alpha})^{n-1}}=\Gamma(1-\alpha)t^{2-\alpha}
\frac{\Gamma\Bigl(1+(n-1)(2-\alpha)\Bigr)}{\Gamma\Bigl(1+n(2-\alpha)\Bigr)}.
$$ 
By Stirling's formula, we have an approximation of the Gamma-function for large values being
$$
\Gamma(1+x)\sim k\sqrt{x}(x/\e )^x
$$ 
for some positive constant $k$. But then we have
\begin{align*}
\frac{\Gamma\Bigl(1+(n-1)(2-\alpha)\Bigr)}{\Gamma\Bigl(1+n(2-\alpha)\Bigr)}&\sim
\sqrt{\frac{n-1}{n}}\frac{\Bigl(\frac{(n-1)(2-\alpha)}{\e}\Bigr)^{(n-1)(2-\alpha)}}{\Bigl(\frac{n(2-\alpha)}{\e}\Bigr)^{n(2-\alpha)}} \\
&=\sqrt{\frac{n-1}{n}}\Bigl(\bigl(1-\frac1n \bigr)^n\frac1{n-1}\frac{\e}{2-\alpha}\Bigr)^{2-\alpha}\rightarrow 0
\end{align*}
when $n\rightarrow\infty$, 
since $(n-1)/n\rightarrow 1$, $(1-1/n)^n\rightarrow \e^{-1}$ and $0<\alpha<1/2$. Hence, 
$H$ is convergent for all $t<\infty$. 
\end{example}

Motivated by these considerations, we state the following claim.
\begin{proposition} \label{existence1}
Consider the generalised Langevin equation \eqref{eq10}, which we recall to be
\begin{equation*}
\ud \xi(t)  = \int\limits_0^t M(t-u)\xi(u) \, \ud u \,\ud t + \chi(t-)  \,\ud L(t),
\end{equation*}
where $\chi$ is a strictly positive $\mathbb{F}$-adapted, c\`adl\`ag process and satisfying $\E[\int_0^T\chi^2(s)\ud s] <\infty$. 
Assume there exists a unique solution $H$ to \eqref{eq02ii}. Further, let $M$ and $H$ be such that 
\begin{equation}\label{Fubini-condition}
\int\limits_0^T \E\Big[\int\limits_0^T M^2(t-u) H^2(u-s) \chi^2(s) \,\ud u \Big]\ud s< \infty  \qquad \forall t\in [0,T].
\end{equation}
Then the analogue of the solution \eqref{eq02}, namely
\begin{equation}\label{eq04}
\xi(t) = H(t) \xi(0) + \int\limits_0^t \chi(u-) H(t-u) \,\ud L(u)
\end{equation}
is the unique solution to the generalised Langevin equation \eqref{eq10}.
\end{proposition}
\begin{proof} The proof consists of two steps. First we show that  $\xi(t)$ in \eqref{eq04} is indeed the unique solution to \eqref{eq10}. In a second step we check that all necessary integrability conditions are satisfied.
	\begin{description}
		\item{\textit{Step 1\ }} Computing the differential of $\xi(t)$ in \eqref{eq04}, we get
\begin{equation}
\ud\xi(t) =\xi(0)\,\ud H(t) +H(0)\chi(t-) \,\ud L(t) + \int\limits_0^t \frac{\ud}{\ud t}H(t-s) \chi(s-) \,\ud L(s) \,\ud t. \label{existence} 
\end{equation}
For two solutions $\xi_1(t)$ and $\xi_2(t)$ of \eqref{existence} with $\xi_1(0)=\xi_2(0)$, it immediately follows from \eqref{existence} that $\xi_1(t)=\xi_2(t)$ for any $t\in  [0,T]$. Thus \eqref{eq04} is the unique solution to \eqref{existence}.
Now, we insert \eqref{eq02ii} in \eqref{existence}, perform a change of variables by putting $u=s+\tau$, and apply the stochastic Fubini theorem, see e.g.\ \cite{Protter}, to arrive at
\begin{align*}
\ud\xi(t) & =  \int\limits_0^tM(t-u)H(u) \,\ud u\,\xi(0) \,\ud t+H(0)\chi(t-) \, \ud L(t) \nonumber\\
&\qquad \qquad + \int\limits_0^t \int\limits_0^{t-s}M(t-s-\tau)H(\tau)\,\ud \tau \chi(s-) \,\ud L(s) \,\ud t\nonumber\\\
& = \int\limits_0^tM(t-u)H(u)\,\ud u\, \xi(0) \,\ud t+H(0)\chi(t-) \,\ud L(t)\nonumber\\\
&\qquad \qquad  + \int\limits_0^t \int\limits_s^t M(t-u)H(u-s)\,\ud u\, \chi(s-) \,\ud L(s)\, \ud t\nonumber\\\
& = \Big[\int\limits_0^tM(t-u)H(u)\xi(0) \,\ud u +  \int\limits_0^t \int\limits_0^t M(t-u) H(u-s)1_{\{s\leq u\}} \,\ud u \,\chi(s-)\,\ud L(s)\Big]\ud t\nonumber\\\
&\qquad \qquad + H(0)\chi(t-)\, \ud L(t) \nonumber\\\
& =  \Big[\int\limits_0^tM(t-u)H(u)\xi(0)\, \ud u +  \int\limits_0^t M(t-u) \int\limits_0^{u} H(u-s)\chi(s-) \, \ud L(s) \,\ud u\Big]\ud t\nonumber\\\
&\qquad \qquad + H(0)\chi(t-) \,\ud L(t) \nonumber\\\
& =\int\limits_0^t M(t-u) \Big[H(u)\xi(0) + \int\limits_0^{u} H(u-s) \chi(s-) \,\ud L(s) \Big] \ud u \,\ud t +H(0) \chi(t-)  \,\ud L(t)\,.\nonumber\
\end{align*}
By observing that the factor between brackets is exactly $\xi(u)$ according to \eqref{eq04}, we conclude that
\eqref{existence} and \eqref{eq10} are equivalent. Since \eqref{eq04} is the unique solution to \eqref{existence} and $H$ satisfying \eqref{eq02ii} is unique, then \eqref{eq04} is also the unique solution to \eqref{eq10}.
\item{\textit{Step 2\ }}
The integrability condition \eqref{Fubini-condition} allows us to apply the stochastic Fubini theorem in Step 1 and implies that the integral terms in equations \eqref{eq04} and \eqref{existence} are well defined. Indeed, from
\begin{align*}
&\int\limits_0^T \E\Big[\int\limits_0^T M^2(t-u) H^2(u-s) \chi^2(s) \,\ud u\Big] \ud s \\
&\qquad = \int\limits_0^T M^2(t-u) \E\Big[\int\limits_0^TH^2(u-s) \chi^2(s) \,\ud s\Big] \ud t < \infty.
\end{align*}
we deduce that $\E\Big[\int_0^TH^2(u-s) \chi^2(s) \,\ud s\Big] < \infty$, guaranteeing in turn that the integral in \eqref{eq04} is well defined. Further, using \eqref{eq02ii} we derive from
\begin{align*}
&\E\Big[\int\limits_0^t \left(\frac{\ud }{\ud t}H(t-s)\right)^2 \chi^2(s) \, \ud s\Big] \\
&\qquad = \E \Big[ \int\limits_0^t \Big(\int\limits_0^{t-s} M(t-s-u)H(u) \,\ud u \Big)^2 \chi^2(s)  \,\ud s\Big]\\
&\qquad \leq T\E\Big[ \int\limits_0^T \int\limits_0^T M^2(t-u) H^2(u-s) \chi^2(s) \,\ud u \,\ud s\Big] < \infty,
\end{align*}
 that also the integral in \eqref{existence} is well defined.
\end{description}
\end{proof}

Let $g$ be a real-valued function on $[0,\infty )$ and $H$ be a differentiable function such that
\begin{align}\label{function-g}
g(t) &= \dot{H}(t),
\end{align}
with $H(0)$ a finite value and
\begin{equation}\label{condition-g}
	\mathbb{E}\Big[\int\limits_0^Tg^2(t-u)\chi^2(u) \,\ud u\Big] <\infty, \quad \forall t\in [0,T].
\end{equation}
Then starting from an SDE of type \eqref{existence}, that is considering 
\begin{align}\label{BSS}
\ud \xi(t) = \int\limits_0^tg(t-u)\chi(u-)  \,\ud L(u) \,\ud t +  g(t)\xi(0)\, \ud t + H(0) \chi(t-) \,\ud L(t),
\end{align}
where $\chi$ is an $\mathbb{F}$-adapted process, we know from the proof of Proposition \ref{existence1}, that the SDE \eqref{BSS} admits a unique solution given by \eqref{eq04}.
Furthermore, we know that if  there is a unique $M$ such that the relation \eqref{eq02ii} holds and $M$ and $H$ satisfy the integrability condition \eqref{Fubini-condition}, then the SDE \eqref{BSS} is equivalent to the SDE \eqref{eq10}. 

SDEs of the type \eqref{BSS} are common for modelling for example temperature and wind speed in energy markets,
see \cite{Benth}. In the next section, we will exploit the relation of these equations to the Langevin equation \eqref{eq10} 
to show the existence of equivalent martingale measures for such processes. To conclude this section we state the link between the Langevin equation \eqref{eq10} and Volterra equations, which is of importance for our analysis later.

\begin{remarks}\label{remarks3}
\begin{description}
\item {(A) } Notice that the SDE \eqref{eq10} under consideration belongs to the class of {\it Volterra equations} driven by a L\'evy process. Those are SDEs of the type
\begin{equation}\label{SDEgeneral}
\ud \xi(t)=a(t,\xi) \, \ud t+b(t-,\xi) \, \ud L(t),\quad \xi(0)=\xi_0,
\end{equation}
where $\xi_0$ is an $\mathcal{F}_0$-measurable random variable satisfying $\mathbb{P}(|\xi_0|<\infty)=1$ and $\xi=(\xi(t))_{0\leq t\leq T}$.
Volterra equations appear naturally in many areas of mathematics such as integral transforms, transport equations, and functional differential equations (we refer to~\cite{GLS} for an introduction and a general overview of these equations in the deterministic case). 
They also appear in applications in biology, physics, and finance. For an example in economics (which also applies to population dynamics), we refer to Example 3.4.1 in \cite{HOUZ}. In the framework of stochastic delay equations and optimal control theory, we refer to~\cite{B} and \cite{OZ}. These processes have recently been proposed in the framework of modelling electricity and commodity prices, see e.g.~\cite{BBV1} and \cite{BBV2}.

The existence and uniqueness of the solution to the SDE of type \eqref{SDEgeneral} is well studied (see e.g., \cite[Theorem 4.6]{Lipster} for Volterra equations driven by Brownian motion and \cite{P} for Volterra equations driven by semimartingales). In our analysis, we showed the existence and uniqueness of the solution by exploiting the link of the Langevin equation \eqref{eq10} to the SDE \eqref{existence}. 

\item{(B) } The results in Proposition \ref{existence1} hold true when $\chi$ is a $p$-dimensional vector process  with $p\geq 1$,  and 
$H: [0,T] \rightarrow \R^{p \times p}$  is a matrix-valued function.  Correspondingly, $M$ and $g$ as defined in \eqref{function-g} will also be matrix-valued functions, i.e., $M: [0,T] \rightarrow \R^{p \times p}$ and $g: [0,T] \rightarrow \R^{p \times p}$. In this case, the solution $\xi$ will be a $p$-dimensional process.
\end{description}
\end{remarks}

\section{Change of measure}\label{change}

Recall the spot market model introduced in Subsection~\ref{subsect:model}, where we assume the conditions of Proposition \ref{existence1} to the
generalised Langevin dynamics \eqref{eq10}.

As recalled in Section \ref{sect1}, we need the dynamics of $S$ under a pricing measure $\mathbb Q$ to price contingent claims on the 
commodity spot $S$. We study a particular class of pricing measures $\mathbb Q$, and 
start by applying the It\^{o} formula to $S=\exp(\xi)$ to derive the dynamics of $S$ under the market probability $\mathbb{P}$ from the dynamics 
of $\xi$ in \eqref{eq10}
\begin{equation}\label{price-process}
\begin{aligned}
	\ud S(t) 
	& =S(t) \left(\int\limits_0^tM(t-u)\log(S(u)) \,\ud u +\chi(t) b + \gamma(t)+ \frac{1}{2}\chi^2(t) c^2\right) \ud t \\
	&\qquad +S(t)\chi(t)  c\,\ud W(t)+S(t-)\int\limits_{\R} \left(\e^{\chi(t-)z} -1\right) \tilde{N}(\ud t, \ud z)\,,
\end{aligned}
\end{equation}
where $\gamma(t) = \int_\R \left(\e^{\chi(t)z} -1-\chi(t) z\right) \ell(\ud z)$. Recall the dynamics of $r$ in \eqref{r-jumps}. Then, the dynamics of the discounted price process $\tilde{S}(t) = \e^{-\int_0^t r(s)\,\ud s} S(t)$, $t\in [0,T]$, is given by
\begin{equation*}
\begin{aligned}
	\ud\tilde{S}(t) 
	 &=\tilde{S}(t) \left(\int\limits_0^t M(t-u)\xi(u) \,\ud u +\chi(t) b + \gamma(t)+ \frac{1}{2}\chi^2(t) c^2-r(t)\right) \ud t \\
	&\qquad  + \tilde{S}(t) \chi(t) c \, \ud W(t)  + \tilde{S}(t-) \int\limits_{\R} \left(\e^{\chi(t-)z} -1\right) \tilde{N}(\ud t, \ud z)\,.
\end{aligned}
\end{equation*}
We consider a pricing measure $\mathbb Q$ defined by $\ud \mathbb{Q} = Z(T)  \, \ud \mathbb{P}$ for a density process $(Z(t))_{0\leq t\leq T}$  (see, e.g., the Girsanov Theorem 1.31 in \cite{OKSU}) 
\begin{equation}\label{process-z-jump}
\begin{aligned}
Z(t) &= \exp\left\{-\int\limits_0^t  \varphi(s,S) \ud W(s) - \frac{1}{2} \int\limits_0^t \varphi^2(s,S) \ud s +\int\limits_0^t\int\limits_\R \log(1-\zeta(s-,z))  \tilde{N}(\ud s, \ud z) \right.\\
&\qquad + \left.\int\limits_0^t\int\limits_\R \left[\log(1-\zeta(s,z)) + \zeta(s,z)\right]\, \ell(\ud z)  \ud s \right\} , \qquad 0\leq t\leq T,
\end{aligned}
\end{equation}
where 
\begin{align} \label{phi}
 \varphi(t,S) &=  \frac{1}{\chi(t) c}\left(\int\limits_0^t M(t-u)\log(S(u))  \,\ud u +\chi(t) b+ \frac{1}{2}\chi^2(t) c^2-r(t)-\rho(t)\right),\\
 \label{zeta}
\zeta(t,z) &=  \frac{\e^{\chi(t) z} -1-\chi(t) z }{\e^{\chi(t) z} -1}= 1-\frac{\chi(t) z}{\e^{\chi(t) z} -1},
\end{align}
with $(\rho(t))_{0\leq t\leq T}$ being the risk premium process defined in \eqref{rho-jumps}.
When the measure $\mathbb{Q}$ exists, then 
\begin{align*}
\ud W_{\mathbb{Q}}(t) &= \varphi(t,S)  \, \ud t + \ud W(t) ,\\
\tilde{N}_{\mathbb{Q}}(\ud t, \ud z) &= \zeta(t,z)  \, \ell(\ud z)  \, \ud t + \tilde{N}(\ud t,\ud z) 
\end{align*}
are respectively a Wiener process and a compensated jump measure under $\mathbb{Q}$ with compensator $\tilde{\ell}_t(\ud z)\ud t = (1-\zeta(t,z)) \ell(\ud z)\ud t$. 
The dynamics of the discounted price process of $\tilde{S}$ under $\mathbb{Q}$ is
\begin{align}\label{dynamics-S-jumps}
\ud \tilde{S}(t) = \tilde{S}(t-)\left(\chi(t) c \,  \ud W_{\mathbb{Q}}(t) + \int\limits_\R \left( \e^{\chi(t-)z} -1\right) \tilde{N}_{\mathbb{Q}}(\ud t, \ud z) +\rho(t)  \, \ud t\right).
\end{align} 
To prove the existence of the measure $\mathbb{Q}$, we need the following lemma
which shows that $(r(t)+\rho(t))^2$ can be bounded by the maximum of 
$(\ln S(t))^2$ on $t\in[0,T]$. To show this crucial bound we apply techniques similar to the proof of
the Gronwall Inequality, that works here due to the particular Ornstein-Uhlenbeck-like structure of the dynamics of $r$ and $\rho$.
\begin{lemma}\label{boundedness-tilde-X}
Let $r$, $\rho$, and $\xi$ be respectively as in \eqref{r-jumps}, \eqref{rho-jumps}, and \eqref{eq10} where we assume $A$, $\bar{A}$, $B_i$, $\bar{B}_i$, $i\in \{1,2,3\}$ to be uniformly bounded by a constant, $B_1$ and $\bar{B}_1$ to be  functions of bounded variation on $[0,T]$, $H(0) = 1$, $\tilde{C}<\chi(t) <C$, $\mathbb{P}$-a.s., $\forall t\in [0,T]$, for strictly positive constants $\tilde{C}$ and $C$, 
and the initial condition $\xi(0) = \xi_0$ is bounded by a constant $\mathbb{P}$-a.s.
 Further, assume that the kernel function $M$ is square integrable over $[0,T]$. Define 
\begin{equation}\label{X}
X(t) = r(t)+\rho(t) .
\end{equation}
Then 
$$
X^2(t) \leq C\left(1+\sup_{0\leq s\leq t}\xi^2(s)\right)\,, \qquad 0\leq t\leq T\,,
$$
for some positive constant $C$ (depending on $T$).
\end{lemma}
\begin{proof}
Inserting \eqref{eq10} and \eqref{Levy} in \eqref{r-jumps}, we find after some rearrangement
\begin{equation}\label{SDEr(t)}
\ud r(t) + B_2(t) r(t) \,\ud t= [A(t) - B_1(t)\chi(t) b ] \,\ud t + B_1(t)\, \ud \xi(t) -B_1(t) \int\limits_0^tM(t-u) \xi(u)\, \ud u \,\ud t.
\end{equation}
Multiplying both sides of \eqref{SDEr(t)} with $\e^{\int_0^tB_2(s)\ud s}$,  applying the product rule to $\ud (\e^{-\int_0^tB_2(u)\,\ud u}B_1(t) \xi(t))$,  integrating both sides from zero to $t$, and dividing by $\e^{\int_0^tB_2(u)\ud u}$ we get
\begin{equation}
\begin{aligned}
r(t)& =  \e^{-\int_0^tB_2(u)\, \ud u}r(0)+\int\limits_0^t\e^{-\int_s^tB_2(u)\,\ud u}[A(s) - B_1(s)\chi(s) b  ] \,\ud s\\
& \phantom{=} + B_1(t)\xi(t)-\e^{-\int_0^tB_2(u)\,\ud u}B_1(0)\xi_0-\int\limits_0^t\e^{-\int_s^tB_2(u)\,\ud u}B_1(s)B_2(s) \xi(s) \,\ud s \\ 
& \phantom{=} -\int\limits_0^t\e^{-\int_s^tB_2(u)\,\ud u}\xi(s)\,\ud B_1(s) -\int\limits_0^t\e^{-\int_s^tB_2(u)\,\ud u}B_1(s)\int\limits_0^sM(s-u) \xi(u)\, \ud u \,\ud s\,.
\end{aligned}
\end{equation}
Applying the triangle inequality, the boundedness of $\chi$, $A$, $B_1$ and $B_2$ and the assumption on $\xi_0$
we bound $r(t)$ for $t\leq T$ by
\begin{align*}
|r(t)| & \leq \e^{K_1T}( K_2+K_3T )+ K_4|\xi(t)| +\e^{K_1T} K_5\int\limits_0^t |\xi(s)|\, \ud s\\
& \phantom{\leq}+ \e^{K_1T}\int\limits_0^t |\xi(s)|| \ud B_1(s)| + \e^{K_1T}K_4 \int\limits_0^t \int\limits_0^s| M(s-u)||\xi(u)|\, \ud u \,\ud s ,
\end{align*}
for some positive constants $K_1,\ldots, K_5$. Since $B_1$ is of bounded variation on $[0,T]$, we further obtain
\begin{align*}
|r(t)| & \leq \e^{K_1T}( K_2+K_3T )+\e^{K_1T} K_6\Big( T+  \int\limits_0^t \int\limits_0^s| M(s-u)|\, \ud u \,\ud s\Big) \sup_{0\leq s\leq t} |\xi(s)|\,,
\end{align*}
 with some additional positive constant $K_6$.
By Cauchy-Schwartz inequality,
\begin{align*}
\int\limits_0^T\int\limits_0^s|M(s-u)|\,\ud u\,\ud s&\leq\int\limits_0^T(\int\limits_0^s1\,\ud s)^{1/2}(\int\limits_0^s|M(s-u)|^2\,\ud u)^{1/2}\ud s \\
&\leq (\int\limits_0^T|M(u)|^2\,\ud u)^{1/2}\int\limits_0^Ts^{1/2}\,\ud s=\frac23 T^{3/2}(\int\limits_0^TM^2(u)\,\ud u)^{1/2}
\end{align*}
which is finite by assumption. Hence, we find,
\begin{equation}
\label{boundr}
|r(t)|\leq \tilde{K}_1+\tilde{K}_2\sup_{0\leq s\leq t}|\xi(s)|,
\end{equation}
where $\tilde{K}_1,\tilde{K}_2$ are two positive constants, depending on $T$.
As $\sqrt{\cdot}$ is an increasing function, 
$$
\sup_{0\leq s\leq t}|\xi(s)|=\sup_{0\leq s\leq t}\sqrt{|\xi(s)|^2}\leq\sqrt{\sup_{0\leq s\leq t}|\xi(s)|^2}
$$
and therefore, by an elementary inequality,
$$
|r(t)|^2\leq C_1+C_2\sup_{0\leq s\leq t}|\xi(s)|^2\,,
$$
for some positive constants $C_1,C_2$ (depending on $T$).

Analogous to  \eqref{r-jumps} we can transform \eqref{rho-jumps} into
\begin{equation}\label{SDErho(t)}
\ud \rho(t)+ \bar{B}_3(t)\rho(t)\,\ud t = [\bar{A}(t) - \bar{B}_1(t)\chi(t) b -\bar{B}_2(t) r(t) ]\, \ud t + \bar{B}_1(t) \,\ud \xi(t)\,,
\end{equation}
Multiplying both sides of \eqref{SDErho(t)}  with  $\e^{\int_0^t\bar{B}_3(s)\ud s}$,  applying the product rule to  $\ud (\e^{-\int_0^t\bar{B}_3(u)\,\ud u}\bar{B}_1(t)\xi(t))$, integrating both sides from zero to $t$, and dividing by  $\e^{\int_0^t\bar{B}_3(u)\ud u}$, we get
\begin{equation}\label{rho(t)solved}
\begin{aligned}
\rho(t)& =  \e^{-\int_0^t\bar{B}_3(u)\, \ud u}\rho(0) +\int\limits_0^t\e^{-\int_s^t\bar{B}_3(u)\,\ud u} [\bar{A}(s) - \bar{B}_1(s)\chi(s) b -\bar{B}_2(s) r(s)]\, \ud s\\ 
&  \phantom{=} + \bar{B}_1(t)\xi(t)-\e^{-\int_0^t\bar{B}_3(u)\,\ud u}\bar{B}_1(0)\xi_0 -\int\limits_0^t\e^{-\int_s^t\bar{B}_3(u)\,\ud u}\bar{B}_1(s)\bar{B}_3(s) \xi(s) \,\ud s\\
&  \phantom{=} -\int\limits_0^t\e^{-\int_s^t\bar{B}_3(u)\,\ud u}\xi(s)\,\ud \bar{B}_1(s) .
\end{aligned}
\end{equation}
Taking absolute values and appealing to the boundedness assumptions along with
  the estimate \eqref{boundr} we can bound $\rho(t)$ in \eqref{rho(t)solved}  for $t\leq T$ by the same arguments as for $r$ to arrive at
\[
|\rho(t)|^2 \leq C_3+ C_4\sup_{0\leq s\leq t} \xi^2(s)\,,
\]
for positive constants $C_3$ and $C_4$ (depending on $T$). The result follows. 
\end{proof}
We state the existence of the martingale measure $\mathbb{Q}$ in the following proposition.

\begin{proposition}\label{martingale-measure-jump}
Under the assumptions of Lemma \ref{boundedness-tilde-X} 
the process $(Z(t))_{0\leq t\leq T}$ defined by \eqref{process-z-jump} is a true $\mathbb{P}$-martingale and $\mathbb{E}[Z(T)]=1$.
\end{proposition}
\begin{proof}
Since $S=\exp (\xi)$, the process $\varphi$ \eqref{phi} can be denoted as a function of $t$ and $\xi$, 
\begin{equation}
\label{hattheta-jump}
\hat{\varphi}(t,\xi):=\varphi(t,S)=\frac{1}{\chi(t)c}[a(t,\xi)+\chi(t) b+ \frac 12 \chi^2(t)c^2-X(t)] .
\end{equation}
where $X$ is given by \eqref{X} and the functions
\begin{equation}\label{a(t,xi)}
a(t,\xi) = \int\limits_0^t M(t-s)\xi(s)\, \ud s \qquad \mbox{and} \qquad b(t,\xi) = \chi(t)
\end{equation}
correspond to those in the version \eqref{SDEgeneral}  of the SDE \eqref{eq10}.
Since by Proposition \ref{existence1} and Remarks \ref{remarks3}, the SDE \eqref{SDEgeneral} has a unique solution $\xi$, the result then follows by applying \cite[Theorem 5.1]{Klebaner} with this semimartingale $\xi$ and the martingale given by
\[
\int\limits_0^t\hat{\varphi}(s,\xi)  \,\ud{W}(s) -\int\limits_0^t \int\limits_\R \zeta(s,z)  \,\tilde{N}(\ud s, \ud z),
\] 
with $\zeta(s,z)$ defined in \eqref{zeta}. 

Hereto we prove that the conditions of \cite[Theorem 5.1]{Klebaner} are satisfied.
\begin{itemize}
	\item $|\xi_0|=|\eta|<C$ by the assumptions in this proposition. 
	\item Invoking H\"older's inequality, the boundedness of $\chi$ and the square integrability of $M$ we obtain for $a(t,\xi)$ in \eqref{a(t,xi)}
	\begin{align}
	|a(t,\xi)|^2  &\leq \Big(\int\limits_0^tM^2(t-s) \,\ud s\Big)\Big( \int\limits_0^t\xi^2(s) \, \ud s\Big) \nonumber\\
&\leq T\int_0^TM^2(s)\,\ud s\sup_{0\leq s\leq t}|\xi(s)|^2\,.\label{bounda}
	\end{align}
Next, the combination of Lemma \ref{boundedness-tilde-X}, equations \eqref{hattheta-jump}, \eqref{bounda} and the uniform boundedness of $\chi$  leads to
	\begin{equation}\label{hatthetabound-jump}
	\hat{\varphi}^2(t,\xi)\leq \frac{3}{\chi^2(t) c^2} \left(a^2(t,\xi) +X^2(t) +[\chi(t)b+ \frac{1}{2}\chi^2(t) c^2]^2\right)\leq C\left(1+\sup_{s\leq t}\xi^2(s)\right) .
	\end{equation}
Let us turn our attention to $\zeta(t,z)$: one can easily show that 
	\[
	|\zeta(t,z)|\leq |\chi(t) z|\,. 
	\]
Indeed consider for the moment the function 
$$
f(u):=\frac{\e^u-1-u}{\e^u-1}\,,
$$
for $u\in\mathbb R$. First, we observe that $\lim_{u\rightarrow 0}f(u)=0$, so the function does not have any singularity at zero. Consider next
$u\geq 0$. Then, since $\exp(u)-1-u=\int_0^u\big(\exp(s)-1\big)\ud s$, we find
$$
f(u)=\int_0^u\frac{\e^s-1}{\e^u-1}\,\ud s\,,
$$ 
which shows that $f(u)\geq 0$ for $u\geq 0$. Moreover, as $\exp(s)-1$ is an increasing function for $s\in [0,u]$, the integrand is less than 1, and we find
$f(u)\leq u$. For $u<0$, we consider the function
$$
g(v)=\frac{1-\e^{-v}-v}{1-\e^{-v}}\,,v>0\,,
$$
where we observe that $g(|u|)=f(u)$ for $u<0$. Then we find similar to above that  $1-\exp(-v)-v=-\int_0^v\big(1-\exp(-s)\big)\ud s$.
Hence,
$$
g(v)=-\int_0^v\frac{1-\e^{-s}}{1-\e^{-v}}\,\ud s\,.
$$
Thus, we see that $g(v)<0$ and that the integrand is non-negative and less than 1 since $1-\exp(-s)$ is increasing for $s\in [0,v]$. It follows that 
$|g(v)|\leq v$ for $v>0$, or, $|f(u)|\leq |u|$. Letting $u:=\chi(t)z$, the desired inequality is reached.
		Hence, using again the uniform boundedness of $\chi$,
	\begin{equation*}
		\int\limits_\R \zeta^2(t,z)  \,\ell(\ud z)\leq 	\int\limits_\R \chi^2(t) z^2 \, \ell(\ud z)\leq \tilde{C}\int\limits_\R z^2  \,\ell(\ud z) ,
	\end{equation*}
	for a positive constant $\tilde{C}$. Adding up the latter and  \eqref{hatthetabound-jump}, we get
	\begin{equation}\label{eq32}
	\hat{\varphi}^2(t,\xi) + \int\limits_\R \zeta^2(t,z) \, \ell(\ud z)  \leq K\left(1+\sup_{s\leq t}\xi^2(s)\right) .
	\end{equation}
	\item 
	Insertion of the decomposition \eqref{Levy} of the process $(L(t))_{0\leq t\leq T}$ transforms the SDE \eqref{SDEgeneral} for $\xi$ into
	$$\ud \xi(t) = (a(t,\xi) +\chi(t) b)  \, \ud t + \chi(t)c  \,\ud W(t) + \chi(t-)\int\limits_\R z   \,\tilde{N}(\ud t, \ud z) .$$
The bound \eqref{bounda} and the uniform boundedness of $\chi$, immediately provide
	\begin{equation}\label{Ltchi}
	L(t,\xi):=(a(t,\xi)+\chi(t) b)^2+\chi^2(t) c^2 + \chi^2(t)\int\limits_\R  z^2 \,\ell(\ud z) \leq C\left(1+\sup_{s\leq t}\xi^2(s)\right).
	\end{equation}
	\item  Noting that $(\e^{\chi(t) z} -1)\chi(t) z\geq 0$ implies $\zeta(t,z)\leq 1$, the uniform boundedness of $\chi$ and the estimates \eqref{eq32} and \eqref{Ltchi} lead to
		\begin{align*}
	\mathfrak{L}(t,\xi)&\!\! := L(t, \xi) + \chi^2(t) c^2 \hat{\varphi}^2(t,\xi) + \chi^2(t) \int\limits_\R  z^2  \, \ell(\ud z)\int\limits_\R \zeta^2(t,z)  \,\ell(\ud z)\\
	&\qquad  + \chi^2(t)\int\limits_\R  z^2\zeta(t,z)  \, \ell(\ud z) \\
	&\leq C\left(1+\sup_{s\leq t}\xi^2(s)\right) .
	\end{align*}
\end{itemize}
Hence the result follows.
\end{proof}
Thus, we may conclude that when the kernel function $M$ is square integrable over $[0,T]$, $|\xi_0| < K$, and $\tilde{C}<\chi(t) <C$, $\mathbb{P}$-a.s., $\forall t\in [0,T]$, for strictly positive constants $K$, $\tilde{C}$ and $C$, and when the functions $B_1$ and $\bar{B}_1$ have bounded variation over $[0,T]$ a measure change exists and the memory in the drift of the dynamics of the asset price under the historical measure $\mathbb{P}$ does not enter the $\mathbb Q$-dynamics (recall \eqref{dynamics-S-jumps}). Recalling example \ref{exampleM} we show that this memory kernel $M$ is a square-integrable memory function singular at zero, to which we can associate an $H$ as in \eqref{def-H-series}
and \eqref{def-rec-b}. Let $M(s)=s^{-\alpha}$ for $\alpha>0$. Then we have that 
$$
\int_0^TM^2(s)ds=\int_0^Ts^{-2\alpha}\,ds=\frac{1}{-2\alpha+1}T^{-2\alpha+1}<\infty
$$
whenever $0<\alpha<1/2$. Further since the corresponding $H$ in \eqref{def-H-series}-\eqref{def-rec-b} is a bounded function in time on $[0,T]$, and the process $\chi$ is assumed to be uniformly bounded in $t$ on $[0,T]$, the square integrability of $M$ guarantees that condition \eqref{Fubini-condition} is satisfied.

\section{Pricing of options and forwards}\label{section5}
In this section we derive analytical formulas for spot option prices and forwards when the underlying price process is modelled as $S = \exp(\xi) $, with  $\xi$ 
satisfying an SDE of the type \eqref{BSS} and which is of interest in energy markets.
\subsection{Example I: L\'evy semi-stationary processes} \label{example1}
\textit{L\'evy semi-stationary processes} (LSS) have been originally introduced in \cite{BBV3} for modelling energy spot prices. This class consists of processes $V$ of the form
\begin{equation}\label{eq61}
V(t)=\int\limits_{-\infty}^t H(t-s)\chi(s-) \, \ud L(s),
\end{equation}
where $H$ is a real-valued function on $[0,\infty)$ with $H(t-s) = 0$ for $s>t$ and $\chi$ is an $\mathbb{F}$-adapted c\`adl\`ag positive process. Here, we have assumed that $L$ is a two-sided L\'evy process, and that the stochastic integral is well-defined.
We note that  $V$ is a Volterra process with a time-homogeneous kernel $H$. 

The name L\'evy semi-stationary processes has been derived from the fact that the process $V$ is stationary as soon as $\chi$ is stationary. In the case that $L$ equals a two-sided Brownian motion $W$, we call such processes Brownian semi-stationary (BSS) processes which have been recently introduced by \cite{BS} in the context of modelling turbulence in physics.

Assume that $H(0)=1$, that relation \eqref{eq02ii} holds for a unique $M$, and that the following integrability condition on $H$ and $\chi$ is satisfied 
$$\int\limits_{-\infty}^T \E\Big[\int\limits_{-\infty}^T M^2(t-u)H^2(u-s) \chi^2(s)  \, \ud u\Big] \ud s < \infty , \quad \forall t \in [0,T] .$$
Computing  the differential of $V$, we arrive at
\begin{align}\label{differential-Z}
\ud V(t) = \int\limits_{-\infty}^t \dot{H}(t-u) \chi(u-) \, \ud L(u) \,\ud t + \chi(t-)  \, \ud L(t).
\end{align} 
Proceeding along the same lines as in the proof of Proposition \ref{existence1}, we obtain that \eqref{eq61} is a unique solution to \eqref{differential-Z}
and 
\begin{align}\label{process-Z}
\ud V(t) &= \int\limits_{-\infty}^t \int\limits_0^{t-u} M(t-u-v) H(v)  \,\ud v \,\chi(u-) \, \ud L(u) \, \ud t + \chi(t-) \, \ud L(t)\nonumber\\
&= \int\limits_{-\infty}^t M(t-s) \int\limits_{-\infty}^s H(s-u)\chi(u-) \, \ud L(u) \, \ud s \, \ud t + \chi(t-) \, \ud L(t) \nonumber \\
&=\int\limits_{-\infty}^t M(t-s) V(s)  \, \ud s \, \ud t + \chi(t-) \, \ud L(t).
\end{align}
\subsubsection{Spot option prices}
From now on, integrals containing the kernel function $M$ have to be considered over $(-\infty,t]$ instead of  $[0,t]$ with $0\leq t\leq T$. In particular, the SDE for the risk premium $\rho$ is 
of OU type   \eqref{rho-jumps} but with mean level $\int_{-\infty}^tM(t-s)V(s)\, \ud s$. 
In the following proposition we state the existence of a martingale measure to compute the option price written on $S= \exp(V)$. This result follows immediately from Proposition \ref{martingale-measure-jump} and equation \eqref{process-Z}.
\begin{proposition} 
Let the assumptions of Proposition \ref{martingale-measure-jump} hold. Then we know there exists a martingale measure $\mathbb{Q}$ under which 
$(\e^{-\int_0^t(\rho(s) +r(s))  \ud s} \e^{V(t)})_ {0\leq t\leq T}$ is a $\mathbb{Q}$-martingale and the price of an (call/put) option written on $S$ is given by 
\begin{equation}
C(t) = \E^\mathbb{Q} \left[\e^{-\int_t^Tr(s)\, \ud s} \max\Big(\varepsilon(S(T)-K), 0\Big)|\mathcal{F}_t\right]\,,\label{option-price}
\end{equation} 
for $t\leq T$, where $\varepsilon=\pm 1$ (for call and put, resp.). Here, $\mathbb{Q}$ is as described in \eqref{process-z-jump} with $\varphi$ given by
$$ \varphi(t,S) =  \frac{1}{\chi(t) c}\left(\,\int\limits_{-\infty}^t M(t-u)\log(S(u))   \ud u +\chi(t) b+ \frac{1}{2}\chi^2(t) c^2-r(t)-\rho(t)\right) $$
and $\zeta$ is as in \eqref{zeta}.
\end{proposition}
Notice that the existence of an equivalent martingale measure for the spot price $S$ was proved  by 
rewriting the process $V$ in terms of the kernel $M$. The proof would not be straightforward if using equation \eqref{eq61} or 
\eqref{differential-Z} directly.

When we assume $r$  and $\chi$ to be deterministic functions of time and exploiting the affine structure under the measure $\mathbb{Q}$ of the process $(Z, \rho)$, we can derive an expression of the option price $C$ in terms of the Fourier transform of the payoff function. For the sake of simplicity we put $\bar{B}_1=\bar{B}_3:=\bar{B}$ in the SDE for $\rho$ but the case $\bar{B}_1\neq \bar{B}_3$ can also be dealt with along similar lines leading to more involved expressions for the solutions of the Riccati equations. 
\begin{proposition} \label{Prop7}
Let the assumptions of Proposition \ref{martingale-measure-jump} hold. Assume $r$ and $\chi$ are deterministic functions of time.  
Define
\begin{equation}\label{tilde-f}
\tilde{f}(\lambda) = \frac{1}{2\pi} \frac{K^{-(\omega -1+\mathrm{i} \lambda)}}{(\omega+\mathrm{i} \lambda)(\omega -1+\mathrm{i} \lambda)}\,, \qquad \omega, \lambda \in \R\,,
\end{equation}
and denote
\begin{align}\nonumber
\varpi_1(t)& = r(t) -\frac{1}{2}\chi^2(t)c^2 -\chi(t) \int\limits_\R z \left(1-\frac{\chi(t) z}{\e^{\chi(t) z}-1}\right)\, \ell(\ud z)\,, \\ \label{varpi2}
 \varpi_2(t) & = \bar{A}(t)+\bar{B}(t)(\varpi_1(t)-\chi(t) b)-\bar{B}_2(t)r(t).
\end{align}
Then the price of an option written on $S$ is given by
\begin{align}\label{option-price1}
C(t) = \e^{-\int_t^Tr(s)\, \ud s} \int\limits_{\R}\e^{\phi(t,T, \,\omega+\mathrm{i}\lambda,\, 0) +  (\omega+\mathrm{i}\lambda) V(t) +(T-t)(\omega+\mathrm{i}\lambda) \rho(t)} \tilde{f}(\lambda) \, \ud \lambda\,,
\end{align}
for
$$
\left\{
    \begin{array}{ll}
        \omega>1\,, & \mbox{if } \quad  \varepsilon=1\,, \quad \mbox{call}, \\
        \omega<0\,, & \mbox{if } \quad  \varepsilon=-1\,, \quad \mbox{put}\,,
    \end{array}
\right.
$$
where the function $\phi$ solves 
\begin{align*}
&\partial_t\phi(t,u_1,u_2)\\
 &\quad = -\frac{1}{2}\chi^2(t)c^2 \left[u_1 + \bar{B}(t) (u_2+u_1(T-t))\right]^2 - \varpi_1(t) u_1  - \varpi_2(t) (u_2+u_1(T-t))\\
&\quad \qquad - \int\limits_\R \Big(\exp\{ [u_1  + \bar{B}(t)(u_2+u_1(T-t)) ]\chi(t) z\}  -1 \\
&\qquad \qquad \qquad 
- [u_1  + \bar{B}(t)(u_2+u_1(T-t)) ]\chi(t) z \Big) \frac{\chi(t) z}{\e^{\chi(t) z}-1}\, \ell(\ud z)\,.
\end{align*}
\end{proposition}
\begin{proof}
Observe that the dynamics of $(V, \rho)$ under the measure $\mathbb{Q}$ is given by 
\begin{align}\label{Z-Q}
\ud V(t) &= \left(\rho(t) +\varpi_1(t) \right) \ud t  + \chi(t) c \, \ud W_\mathbb{Q}(t)+ \chi(t)\int\limits_\R  z \tilde{N}_\mathbb{Q}(\ud t, \ud z)\,, \\ \label{rho-Q}
\ud \rho(t) &= \varpi_2(t)\ud t + \bar{B}(t)\chi(t) c\, \ud W_\mathbb{Q}(t) + \bar{B}(t)\chi(t) \int\limits_\R z  \tilde{N}_\mathbb{Q}(\ud t, \ud z)\,.
\end{align}
Define $\varpi^{\top}(t)  = (\varpi _1(t), \varpi _2(t))$,
$$\varrho(t) =  \chi^2(t) c^2\begin{bmatrix}
    1   &   \bar{B}(t)  \\
        \bar{B} (t)     &   \bar{B}^2(t)\,
\end{bmatrix},
\quad 
\beta=  \begin{bmatrix}
0     & 1  \\
          0     &  0\,
\end{bmatrix}
\quad
\iota(t,z) = \chi(t) z \begin{bmatrix} 1\\
\bar{B}(t) \,
\end{bmatrix}\,, 
$$
for $z \in \R$. Let  $(\phi, \psi)$ be a solution to the following Riccati equations 
\begin{align*}
\partial_t \phi(t,T,u) &= - \frac{1}{2} \psi(t,T,u)^{\top} \varrho (t)\, \psi(t,T,u) - \varpi^{\top}(t)\psi(t,T,u) \\
&\qquad- \int\limits_\R \left(\e^{\psi(t,T,u)^{\top}\iota(t,z)} -1- \psi(t,T,u)^{\top} \iota(t,z)\right) \tilde{\ell}_t(\ud z) \,,\\
\phi(T,T,u) &= 0\,,\\
\partial_t \psi(t,T,u) & = -\beta^{\top} \psi(t,T,u)\,, \\
\psi(T,T,u) &= u\,, 
\end{align*}
for $0\leq t \leq T$ and $u \in \mathbb{C}^2$, where we recall that $\tilde{\ell}_t(\ud z)\ud t$ is the compensator of ${N}(\ud t, \ud z)$ under ${\mathbb{Q}}$. Then there exists a unique global solution $(\phi(\cdot,T,u), \psi(\cdot,T,u)): [0,T] \rightarrow \mathbb{C}\times \mathbb{C}^2$\,, for $u \in \mathbb{C}^2$ to the latter system given by 
\begin{align*}
\phi(t,T,u) &= \int\limits_t^T\frac{1}{2} \chi^2(s) c^2\left( \psi_1(s,T,u) +  \bar{B}(s) \psi_2(s,T,u)\right)^2 \, \ud s \\
&\qquad +\int\limits_t^T (\varpi_1(s) \psi_1(s,T,u) + \varpi_2(s) \psi_2(s,T,u))\, \ud s\\
&\qquad + \int\limits_t^T\int\limits_\R \Big(\exp\{[\psi_1(s,T,u)   + \bar{B}(s)\psi_2(s,T,u)]\chi(s) z \} -1 \\
&\qquad \qquad - [\psi_1(s,T,u)   + \bar{B}(s)\psi_2(s,T,u)]\chi(s) z\Big) \, \frac{\chi(s) z}{\e^{\chi(s) z}-1}\,\ell(\ud z) \, \ud s\,,\\
\psi_1(t,T,u) &= u_1\,,\\
\psi_2(t,T,u) &= u_2+u_1(T-t)\,,
\end{align*}
where we used the fact that $\tilde{\ell}_t(\ud z) = (1-\zeta(t,z))\,\ell(\ud z)$, with $\zeta(t,z)$ given by \eqref{zeta}. 
We fix $\omega >1$ and $u^{\top} = (\omega+\mathrm{i}\lambda, 0)$. The latter together with Theorem \ref{Affine-process-charac}, an adaptation of \cite[Theorem 3.4]{KM} to our setting, and \cite[Lemma 10.2]{Filipovic} yield \eqref{option-price} for $\varepsilon = 1$, where $\tilde{f}$ is as in \eqref{tilde-f}. We fix $\omega<0$ and $u^{\top} = (\omega+\mathrm{i}\lambda, 0)$, then also \eqref{option-price} holds for $\varepsilon=-1$. Thus the statement of the proposition follows.
\end{proof}

\subsubsection{Forward prices}\label{Forwards-continuous}
We consider a forward contract on $S$, contracted at $t$ with time of delivery $T\geq t$ and with the forward price $F(t,T)$ such that
$$\E^{\mathbb{Q}}\left[\e^{-\int_t^T r(s) \,\ud s} \left(F(t,T) -S(T)\right) | \mathcal{F}_t\right] =0\,. $$
The latter is equivalent to 
\begin{equation}\label{forward}
F(t,T) = \left(\E^{\mathbb{Q}}\left[\e^{-\int_t^T r(s) \,\ud s}|\mathcal{F}_t\right]\right)^{-1}\E^{\mathbb{Q}}\left[\e^{-\int_t^T r(s) \,\ud s} S(T) | \mathcal{F}_t\right] \,.
\end{equation}
Assuming $r$ and $\chi$ are deterministic, we state in the following proposition an analytical expression for the latter forward price \eqref{forward}.
\begin{proposition}
Let the assumptions of Proposition \ref{martingale-measure-jump} hold. Moreover assume that $r$ and $\chi$ are deterministic functions of time.
Then the forward price \eqref{forward} is given by
\begin{equation}\label{forward-price}
F(t,T) = S(t)\exp\left(\mathcal{A}(t,T)+\int\limits_t^T r(s)\, \ud s +(T-t)\rho(t)\right)
\end{equation}
with
\begin{equation}\label{mathcalA}
\begin{aligned}
\mathcal{A}(t,T)& =\frac 12 c^2 \int\limits_t^T\bar{B}^2(s)\chi^2(s)(T-s)^2\,\ud s +\int\limits_t^T (T-s) \Big[\bar{A}(s)+(\bar{B}(s)-\bar{B}_2(s))r(s)\Big] \ud s\\
&\qquad -\left[b +\int\limits_\R  z \, \ell(\ud z)\right]\int\limits_t^T\bar{B}(s)\chi(s)(T-s)\,\ud s - \frac 32 c^2\int\limits_t^T\bar{B}(s)\chi^2(s)(T-s)\,\ud s \\
 &\qquad + \int\limits_t^T \int\limits_\R \frac{\chi(s) z\e^{\chi(s) z}}{\e^{\chi(s) z} -1} \left[\e^{ (T-s)\bar{B}(s)\chi(s) z} - 1 \right]\,\ell(\ud z) \, \ud s   \,.
\end{aligned}
\end{equation}
\end{proposition}
\begin{proof}
The process $\hat{Z}$ defined by
$$\hat{Z}(t) := \frac{S(t)}{S(0)}\e^{-\int_0^t (r(s)+\rho(s))\, \ud s} 
\,, \qquad 0\leq t\leq T,$$ 
is a positive $\mathbb{Q}$-martingale with $\mathbb{E}^{\mathbb{Q}}[\hat{Z}(t)]=1$ and defines a probability measure $\mathbb{\hat{Q}} \sim \mathbb{Q} $ by
$$\frac{\ud \mathbb{\hat{Q}}}{\ud \mathbb{Q}}\big|_{\mathcal{F}_t} = \hat{Z}(t)\,, \qquad 0\leq t\leq T.$$
Considering a deterministic $r$ and applying Bayes' rule in \eqref{forward}, we get
\begin{align}\label{bayes}
F(t,T) = S(t) \e^{\int_t^Tr(s) \,\ud s } \E^{\mathbb{\hat{Q}}}\left[\exp\Big(\int\limits_t^T \rho(s) \, \ud s\Big) \Big|  \mathcal{F}_t\right]\,.
\end{align}
Moreover from the Girsanov theorem (see e.g.~\cite[Theorem 1.31]{OKSU}), we know that
\begin{align*}
W_{\mathbb{\hat{Q}}}(t) &= W_{\mathbb{Q}}(t) + \chi(t) c t\,,\\
\tilde{N}_{\mathbb{\hat{Q}}}(\ud t, \ud z) &= (1-\e^{\chi(t) z})\, \tilde{\ell}_t(\ud z) \ud t + \tilde{N}_{\mathbb{Q}}(\ud t, \ud z)\,.
\end{align*}
Recall the dynamics of $\rho$ under $\mathbb{Q}$ in \eqref{rho-Q}. Then the dynamics of $\rho$ under $\mathbb{\hat{Q}}$ is given by
\begin{align*}
\ud \rho(t) &= \Big[\varpi_2(t) - \bar{B}(t)\chi^2(t)c^2 - \bar{B}(t)\chi(t)\int\limits_\R z(1-\e^{\chi(t) z})\, \tilde{\ell}_t(\ud z)\Big] \ud t + \bar{B}(t) \chi(t) c \, \ud W_{\hat{\mathbb{Q}}}(t)\\
&\qquad 
+ \bar{B}(t)\chi(t) \int\limits_\R z \, \tilde{N}_\mathbb{\hat{Q}}(\ud t, \ud z)\,.
\end{align*}
Applying Theorem \ref{Affine-process-charac}, the expectation in \eqref{bayes} becomes
$$\E^{\mathbb{\hat{Q}}}\left[\exp\Big(\int\limits_t^T \rho(s) \, \ud s\Big) \Big|  \mathcal{F}_t\right] = \exp\Big(\mathcal{A}(t,T)  + \mathcal{B}(t,T) \rho(t)\Big)\,,$$
where
\begin{align*}
\mathcal{B}(t,T) &= T-t\,,\\
\mathcal{A}(t,T) &= \frac 12 c^2 \int\limits_t^T\bar{B}^2(s)\chi^2(s)(T-s)^2\,\ud s -c^2\int\limits_t^T\bar{B}(s)\chi^2(s)(T-s)\,\ud s\\
&\quad  - \int\limits_t^T\bar{B}(s)\chi(s)(T-s)
\int\limits_\R  z(1-\e^{\chi(s) z}) \, \tilde{\ell}_s(\ud z) \ud s + \int\limits_t^T\varpi_2(s)(T-s)\, \ud s  \\
&\quad  + \int\limits_t^T\int\limits_\R\Big[\e^{(T-s)\bar{B}(s)\chi(s) z} -1 -(T-s) \bar{B}(s)\chi(s) z\Big]\,\hat{\ell}_s(\ud z)\, \ud s\,,
\end{align*}
where $\hat{\ell}_t(\ud z) \ud t= \e^{\chi(t) z}\, \tilde{\ell}_t(\ud z)\ud t=\e^{\chi(t) z}(1-\zeta(t,z))\ell(dz)\ud t$ is the compensator of ${N}(\ud t, \ud z)$ under ${\mathbb{\hat{Q}}}$. 
Substituting the expression for $\varpi_2(t)$ from \eqref{varpi2} and $\zeta(t,z)$ from \eqref{zeta}, combining and simplifying all these expressions in \eqref{bayes}, we get the result.
\end{proof}

\subsubsection{Options on forwards} 
In this subsection, we consider a call option with strike $K$ and exercise time $T$ written on the forward $F(t,T)$. 
Observe that the price at time $t\leq T$ of this option under the pricing measure $\mathbb{Q}$ defined by  \eqref{process-z-jump} will be 
$$P(t) = \E^{\mathbb{Q}}\left[\e^{-\int_t^T r(s)\, \ud s }(F(t,T)-K)^+ |\mathcal{F}_t\right],$$
where  $F$ is given by \eqref{forward-price}. 

In the following proposition, assuming that $r$ is deterministic, we derive an expression of the call option written on $F$ in terms of the Fourier transform of the payoff function.
\begin{proposition}\label{option-price-forward}
Let the assumptions of Proposition \ref{martingale-measure-jump} hold. Assume that $r$ and $\chi$ are deterministic functions of time. Recall the expression of $\tilde{f}$ in \eqref{tilde-f}.
Then the call option written on $F$ is given by
$$P(t) = \e^{-\int_t^T r(s)\, \ud s } \int\limits_{\R}\e^{\phi(t,T, \,\omega+\mathrm{i}\lambda) + (\omega+\mathrm{i}\lambda) \log F(t,T)} \tilde{f}(\lambda) \, \ud \lambda\,, \qquad \mbox{for }  \omega >1\,,$$
where the function $\phi$ solves 
\begin{align*}
&\partial_t\phi(t,T,u) = -\frac{1}{2} c^2\Sigma^2(t,T) u^2 - \tilde{A}(t,T)u - \int\limits_\R\left(\e^{\Sigma(s,t)zu} -1 - \Sigma(s,t)zu \right)\frac{\chi(t) z}{\e^{\chi(t) z}-1} \ell(\ud z)
\end{align*}
and the process $(\log F(t,T))_{0\leq t\leq T}$ satisfies the SDE
\begin{align}\label{gamma-jumps}
\ud \log F(t,T) &= \tilde{A}(t,T) \, \ud t   + c\Sigma(t,T) \, \ud W_{\mathbb{Q}}(t) + \Sigma(t,T)\int\limits_\R z\, \tilde{N}_\mathbb{Q}(\ud t, \ud z)\,, 
\end{align}
with
\begin{align*}
\tilde{A}(t,T) &= -\frac{1}{2} (T-t)^2 \bar{B}^2(t) \chi^2(t)c^2 + (T-t)\bar{B}(t)\chi^2(t)c^2 - \frac{1}{2} \chi^2(t)c^2 \nonumber\\
&\phantom{=}  
+\int\limits_\R \chi(t) z\left(\frac{\Sigma(t,T)z}{\e^{\chi(t) z}-1}-1\right) \ell(\ud z) 
+  \int\limits_\R \frac{\chi(t) z\e^{\chi(t) z}}{\e^{\chi(t) z} -1} \left[\e^{ (T-t)\bar{B}(t)\chi(t) z} - 1 \right]\ell(\ud z)  \,, \nonumber\\
\Sigma(t,T) &=  \Big(1 + \bar{B}(t)(T-t)\Big)\chi(t)\,.
\end{align*}
\end{proposition}
\begin{proof}
Recall that $S=\e^V$, and thus the expression \eqref{forward-price} can be rewritten as 
	\[
F(t,T)=\e^{\Upsilon(t,T)}\;\; \mbox{with}\;\;
\Upsilon(t,T)=V(t)+\mathcal{A}(t,T)+\int\limits_t^T r(s)\, \ud s +\rho(t)(T-t).
\]
Differentiating $\Upsilon(t,T)$ with respect to $t$ and substituting the dynamics \eqref{Z-Q}-\eqref{rho-Q} of $V$ and $\rho$ under $\mathbb{Q}$ and the expression \eqref{mathcalA} for $\mathcal{A}(t,T)$ we get the dynamics \eqref{gamma-jumps} for $\log F(t,T)=\Upsilon(t,T)$.
Let $\phi$ and $\psi $ be a solution to the following Riccati equations 
\begin{align*}
\partial_t \phi(t,T,u) &= -\frac{1}{2} c^2\Sigma^2(t,T)\psi^2(t,T,u)  - \tilde{A}(t,T) \psi(t,T,u)\\
&\qquad  - \int\limits_\R \Big(\e^{\psi(t,T,u) \Sigma(t,T)z} - 1- \psi(t,T,u)\Sigma(t,T)z\Big)(1-\zeta(t,z)) \ell(\ud z)  \,,\\
\phi(T,T,u) &= 0\,,\\
\partial_t \psi(t,T,u) & = 0\,, \\
\psi(T,T,u) &= u\,, 
\end{align*}
for $(t,u) \in \R_+\times \mathbb{C}$ and where $\zeta(z)$ is defined in \eqref{zeta}.  
Then there exists a unique global solution $(\phi(\cdot,u), \psi(\cdot,u)): \R_+ \rightarrow \mathbb{C}\times \mathbb{C}$\,, for all $u \in \mathbb{C}$ to the latter system given by 
\begin{align*}
\phi(t,T,u) &= \int\limits_t^T \left(\frac{1}{2}c^2 \Sigma^2(s,T) \psi^2(s,T,u) + \tilde{A}(s,T) \psi(s,T,u) \right) \ud s\\
&\quad + \int\limits_t^T\int\limits_\R\Big(\e^{\psi(s,T,u)\Sigma(s,T)z } -1 - \psi(s,T,u)\Sigma(s,T)z\Big)(1-\zeta(t,z))  \ell(\ud z) \, \ud s\,,\\
\psi(t,T,u) &= u\,.
\end{align*}
We fix $\omega >1$ and $u = \omega+\mathrm{i}\lambda$. The latter together with Theorem \ref{Affine-process-charac}, an adaptation of \cite[Theorem 3.4]{KM} to our setting, and \cite[Lemma 10.2]{Filipovic} yield
the statement of the proposition. 
\end{proof}

\subsection{Example II: CARMA processes}

In this section we introduce a model that is used for modelling temperature and wind speed in weather derivatives, see \cite[\S 4]{Benth}. It is the continuous-time version of an ARMA model, namely a CARMA model where CARMA stands for {\it continuous-time autoregressive moving average}. For the special case of a CARMA(2,1) model one of the factors determining the temperature or wind speed satisfies an SDE of type \eqref{eq10}. As noted in \cite[\S 6.2]{Benth} stationary CARMA processes are special cases of BSS processes.

Let us introduce the model as in \cite{Benth}.
Consider the stochastic vector process $\textbf{X}$ taking values in $\mathbb{R}^p$, with $p\geq 1$, which is defined as the solution to the SDE
\begin{equation}\label{eq20}
\ud\textbf{X}(t)= A\textbf{X}(t)\ud t +\textbf{e}_p\vartheta(t) \ud W(t)
\end{equation}
where $\textbf{e}_k \in\mathbb{R}^p$, $k = 1, \ldots,p$, is the $k$th standard Euclidean basis vector of $\mathbb{R}^p$. The $(p \times p)$-matrix $A$ is defined as
\begin{equation}\label{matrixA}
\begin{bmatrix}
0 & 1 & 0 & \cdots & 0\\
0 & 0 & 1 & \cdots & 0\\
\vdots & \vdots & \vdots& \ddots & \vdots \\
0 & 0 & 0 & \cdots & 1\\
-\alpha_p & -\alpha_{p-1} & -\alpha_{p-2} & \cdots & -\alpha_1
\end{bmatrix} .
\end{equation}
The constants $\alpha_k$, $k = 1, \ldots,p$, are assumed to be non-negative with $\alpha_p > 0$. 
The deterministic function $\vartheta : \mathbb{R}^+ \rightarrow \mathbb{R}$ in the dynamics of $\textbf{X}$ is assumed to be bounded, continuous and
$\vartheta(t)$ is strictly bounded away from zero, i.e., there exists a constant $\overline{\vartheta} > 0$ such that $\vartheta(t)\geq \overline{\vartheta}$ for all $t \geq  0$. Further, we will assume for the sake of simplicity that $\textbf{X}(0)=\textbf{0}$.

Since the entries in the matrix $A$ are constants it follows from \cite[Theorem 4.10]{Lipster} 
 that the unique solution to the SDE \eqref{eq20} with initial condition $\textbf{X}(0)=\textbf{0}$ is given by
\begin{equation}\label{eq21}
\textbf{X}(t)=\int\limits_0^t \exp(A(t-s))\textbf{e}_p\vartheta(s) \ud W(s).
\end{equation}
Differentiating ${\textbf{X}}$ in \eqref{eq21}, we get 
\begin{align}\label{SDEX}
\ud {\textbf X}(t) &= \textbf{e}_p \vartheta(t)\ud W(t) + \int\limits_0^t A\e^{A(t-s)}\textbf{e}_p\vartheta(s)\ud W(s) \ud t.
\end{align}
The latter is an SDE of type \eqref{BSS} with 
$$g(t) = A\e^{At}, \qquad  \sigma(t) =\textbf{e}_p \vartheta(t), \qquad  {\textbf X}(0) = \textbf{0}.$$
From the existence of the solution \eqref{eq21} we infer that the integrability condition \eqref{condition-g} holds and we have 
\begin{equation}\label{finitness}
 \int\limits_0^T \left( A\exp(A(t-s))\textbf{e}_p\right)^{\top}A\exp(A(t-s))\textbf{e}_p\vartheta^2(s) \ud s < \infty.
\end{equation}
From relation \eqref{eq02ii}, the kernel $M$ associated with the process $\textbf{X}$ is such that
\begin{equation}
\label{eq32bis} 
g(t) = A\e^{At} = \int\limits_0^tM(t-u) H(u) \ud u = \int\limits_0^tM(t-u) \e^{Au}  \ud u.
\end{equation}
A solution to the latter equation is given by $M(t)= A \delta_0(t)$, where $\delta_0$ is the Dirac delta function at zero. Using Proposition \ref{existence1}, we know $\textbf{X}$ can be rewritten as  
\begin{align*}
\ud {\textbf X}(t) &= \textbf{e}_p \vartheta(t) \ud W(t) + \int\limits_0^tA \delta_0(t-u) {\textbf X}(u) \ud u \ud t\\
&= \textbf{e}_p \vartheta(t)\ud W(t) + A{\textbf X}(t) \ud t,
\end{align*}
and we recover equation \eqref{eq20} for $\textbf{X}$. Of course in this particular case this latter result follows straightforwardly from \eqref{SDEX}. 
Here we wanted to point out that the process ${\textbf X}$ fits into our framework. 

For $0 \leq q < p$, define the vector $\textbf{b} \in \mathbb{R}^p$ with coefficients $b_j$, $j =
0, 1, \ldots , p -1$ satisfying $b_q = 1$ and $b_j = 0$ for $q < j < p$. 
We define the CARMA$(p, q)$ process $Y$ as
\begin{equation}\label{eq22}
Y (t) = {\textbf{b}}^{\top}\textbf{X}(t)\,.
\end{equation}
Substituting the expression \eqref{eq21} for $\textbf{X}(t)$, we can express $Y$ as
\begin{equation}\label{eq23}
Y (t) = \int\limits_0^t {\textbf{b}}^{\top}\exp(A(t-s))\textbf{e}_p\vartheta(s) \ud W(s) .
\end{equation}
It is evident that $Y$ is a Gaussian process with mean zero and variance 
$$\int\limits_0^t({\textbf{b}}^{\top}\exp(A(t-s))\textbf{e}_p)^2\vartheta^2(s) \ud s.$$
There are two special cases of the CARMA$(p,q)$. First, when $\textbf{b}=\textbf{e}_1$, then $q=0$ and we obtain a so-called CAR$(p)$ process that is used to model temperature. Second, requiring $q=p-1$ then we get for example for $p=2$, a CARMA$(2,1)$ model that can be used to model wind speed.

Applying the It\^{o} formula to \eqref{eq23}, we find as dynamics of the process $Y$ 
\begin{equation}\label{eq24}
\ud Y (t) = \left(\int\limits_0^t {\textbf{b}}^{\top} A\exp(A(t-s))\textbf{e}_p\vartheta(s) \ud W(s)\right) \ud t + {\textbf{b}}^{\top} \textbf{e}_p\vartheta(t) \ud W(t).
\end{equation}
Clearly, the drift term involves the whole path of the Wiener process up to time $t\geq 0$. 

The SDE \eqref{eq24} for $Y$ is also of type \eqref{BSS} with 
$$g(t) = \textbf{b}^{\top} A\exp(At)\textbf{e}_p, \quad  H(t)=\textbf{b}^{\top} \exp(At)\textbf{e}_p, \quad \sigma(t) = \vartheta(t), \quad {Y}(0) = 0.$$
To fit into our framework, we need that the drift in \eqref{eq24} can be written
as $\int_0^tM(t-s)Y(s)\,\ud s$ for some memory function $M$.
Substituting \eqref{eq21} in \eqref{eq24} we have 
\begin{equation*}
\ud Y (t) = {\textbf{b}}^{\top} A {\textbf X}(t) \ud t + {\textbf{b}}^{\top} \textbf{e}_p\vartheta(t) \ud W(t)\,,
\end{equation*}
and hence $M$ must satisfy the equation
\begin{equation}
\label{carma-m}
\mathbf{b}^{\top}A\mathbf X(t)=\int_0^tM(t-s)Y(s)\,\ud s\,.
\end{equation}
We emphasize that $M$ is a scalar function here. Since $Y(t)=\mathbf{b}^{\top}\mathbf X(t)$, we find
$$
\mathbf{b}^{\top}A\mathbf X(t)=\mathbf{b}^{\top}\int_0^tM(t-s)\mathbf X(s)\,\ud s\,.
$$
Take the Laplace transform on both sides to find
$$
\mathbf b^{\top} A \hat{\mathbf X}(\theta)=\mathbf b^{\top}(\hat{M}(\theta)\hat{\mathbf X}(\theta))\,,
$$
for all $\theta\in\mathbb R$, where $\hat{M}(\theta)$ is the Laplace transform \eqref{LaplaceM} of $M$ and $\hat{\mathbf X}(\theta)$ is a vector
in $\mathbb R^p$ consisting of the Laplace transform of the coordinates of $\mathbf X$. But then we see that $\hat{M}(\theta)$ is an eigenvalue
of $A$ for all $\theta\in\mathbb R$ with eigenvector $\hat{\mathbf X}(\theta)$. As $A$ has fixed eigenvectors and eigenvalues, this cannot be true and we conclude that there exist no $M$ such that \eqref{carma-m} holds. 

Now we cannot apply the results in Section 4 for the measure change. An alternative argument to show the existence of a measure change for
a process being the exponential of a CARMA is presented in \cite[Proposition 5.1]{Benth}. We will extend this result and study pricing in
our context.

To this end, for a CARMA($p, q$)-process $Y(t)=\textbf{b}^{\top}\textbf{X}(t)$, define the spot dynamics as
\begin{equation}\label{SofY}
S(t)=\exp(\mu t+Y(t))\,,
\end{equation}
where we suppose that $\textbf{X}\in\mathbb R^p$ follows the Ornstein-Uhlenbeck dynamics with level $\boldsymbol{\xi}\in\mathbb R^p$,
$$
\ud \textbf{X}(t)=(\boldsymbol{\xi}+A\textbf{X}(t))\,\ud t+\vartheta\textbf{e}_p\, \ud W(t)\,.
$$
A simple application of the multidimensional It\^o formula gives the $\mathbb{P}$-dynamics of $S$ as
\begin{equation}
\ud S(t)=\left(\mu+\frac12(\textbf{b}^{\top}\textbf{e}_p)^2\vartheta^2+\textbf{b}^{\top}(\boldsymbol{\xi}+A\textbf{X}(t))\right)S(t)\, \ud t+\vartheta(\textbf{b}^{\top}\textbf{e}_p)S(t)\, \ud W(t)\,.
\end{equation}
For simplicity we suppose that $\vartheta$ is constant. 

We recall that $\textbf{b}^{\top}=(b_0,b_1,\ldots,b_q,0\ldots,0)\in\mathbb R^p$, where $q<p$ and
$b_q=1$. If $q<p-1$, then $\textbf{b}^{\top}\textbf{e}_p=0$ since the last coordinate of $\textbf{b}$ in
this case is zero. Hence, if $q<p-1$ we do not have any martingale term in the dynamics of $S$ and therefore
cannot change measure to any $\mathbb{Q}\sim \mathbb{P}$ so that $\exp(-rt)S(t)$ is a $\mathbb{Q}$-martingale. We assume from now on
that
$$
q=p-1\,.
$$
In this case, $\textbf{b}^{\top}\textbf{e}_p=1$ and the dynamics of $S$ under $\mathbb{P}$ is
\begin{equation}\label{S-CARMA}
\ud S(t)=\left(\mu+\frac12\vartheta^2+\textbf{b}^{\top}(\boldsymbol{\xi}+A\textbf{X}(t))\right)S(t)\, \ud t+\vartheta S(t)\, \ud W(t)\,.
\end{equation}
Introduce next a state-dependent risk premium, 
\begin{equation}
\label{convenience-spec-carma}
\rho(t):=\textbf{c}^{\top}(\boldsymbol{\theta}+C\textbf{X}(t))\,,
\end{equation}
where $\boldsymbol{\theta}\in\mathbb R^p$, 
$\textbf{c}^{\top}=(c_0,c_1,\ldots,c_{\tilde{q}},0,\ldots,0)\in\mathbb R^p$, $\tilde{q}<p$ and
\begin{equation}\label{Cmatrix}
C=\begin{bmatrix}
0 & 1 & 0 & \cdots & 0\\
0 & 0 & 1 & \cdots & 0\\
\vdots & \vdots & \vdots& \ddots & \vdots \\
0 & 0 & 0 & \cdots & 1\\
-\beta_p & -\beta_{p-1} & -\beta_{p-2} & \cdots & -\beta_1
\end{bmatrix}\,,
\end{equation}
for $\beta_i>0, i=1,\ldots,p$. The natural choice of ``moving-average vector'' $\textbf{c}$ is, as we will see,
$\textbf{c}=\textbf{b}$.

Define the stochastic process
\begin{equation}
\theta(t, \textbf{X})=\vartheta^{-1}\left(\mu+\frac12\vartheta^2-r+\textbf{b}^{\top}({\boldsymbol{\xi}}+A\textbf{X}(t))-\rho(t)\right)\,.
\end{equation}
With the specification of the risk premium as in \eqref{convenience-spec-carma}, we find
\begin{equation}
\label{theta-carma}
\theta(t, \textbf{X})=\frac{\mu+\frac12\vartheta^2-r+\textbf{b}^{\top}\boldsymbol{\xi}-\textbf{c}^{\top}\boldsymbol{\theta}}{\vartheta}+\vartheta^{-1}
\left(\textbf{b}^{\top}A-\textbf{c}^{\top}C\right)\textbf{X}(t)\,.
\end{equation}
Hence, $\theta$ is a constant plus a linear combination of the coordinates of $\textbf{X}(t)$. 

In the following proposition we state the existence of a martingale measure to compute the option price written on $S$ as given in \eqref{S-CARMA}.
\begin{proposition}\label{Prop10}
Let $S$ be as in \eqref{S-CARMA} with $\textbf{X}$ defined by the SDE \eqref{eq20} and let $q=p-1$. Assume that the eigenvalues of the matrix $A$ all have negative real part, the risk-free interest rate $r$ is constant and $\rho$ is defined by \eqref{convenience-spec-carma}. Then there exists a martingale measure $\mathbb{Q}$ under which $(\e^{-
\int_0^t\rho(s)\ud s - rt} S(t))_{0\leq t \leq T}$ is a $\mathbb{Q}$-martingale  and the price of an (call/put)
option written on $S$ is given by
$$C(t)=\e^{-r(T-t)}\mathbb{E}^{\mathbb{Q}}\Big[ \max\Big(\varepsilon(S(T)-K),0\Big)\mid \mathcal{F}_{t}\Big],$$
for $t\leq T$ and $\varepsilon=\pm 1$ (call and put, resp.), where $\mathbb{Q}$ is given by 
\begin{equation*}
\frac{\ud \mathbb{Q}}{\ud \mathbb{P}}\Big\vert_{\mathcal{F}_t} =
 Z(t)=\exp\Big[\int\limits_0^t\theta(s,{\textbf{X}}) \ud{W}(s)-\frac 12 \int\limits_0^t \theta^2(s,{\textbf{X}}) \ud s\Big], \quad t\in [0,T],
\end{equation*}
$\theta(t,\textbf{X})$ is as in \eqref{theta-carma}, and the $\mathbb{Q}$-dynamics of $S$ is given by
\begin{equation}\label{SDEforSunderQ}
\ud S(t)=(r+\rho(t))S(t)\,\ud t+\vartheta S(t)\, \ud W_\mathbb{Q}(t)\,,
\end{equation}
for a $\mathbb Q$-Brownian motion $W_\mathbb{Q}$.
\end{proposition}
\begin{proof}
The result follows immediately by extending the result of \cite[Proposition 5.1]{Benth}\footnote{The proof of this extension requires a reasoning using the spectral representation applied to the function $g_{\xi}(u)=\theta'\exp(Au)\boldsymbol{\xi}$ similarly as applied to the function $g(u)=\theta'\exp(Au) \textbf{e}_p$ in the proof of \cite[Proposition 5.1]{Benth}.} to the case of a non-zero mean level $\boldsymbol{\xi}$ where the measure change as described in this proposition takes the form of  
\[
\theta_0(s)=\mu+\frac12\vartheta^2-r+\textbf{b}^{\top}\boldsymbol{\xi}-\textbf{c}^{\top}\boldsymbol{\theta}\,,
\quad \theta'=\textbf{b}^{\top}A-\textbf{c}^{\top}C\,,     \quad \sigma(s)=\vartheta \,.
\]
\end{proof}
Note that for the measure change to be structure preserving the components $\theta_i$ in the vector $\theta'$ should satisfy
\[
\theta_i<\alpha_i, \; i=1,\ldots, p,
\]
as it is assumed in \cite[Proposition 5.1]{Benth}.


By definition,
$\textbf{X}(t)$ is a $p$-variate Ornstein-Uhlenbeck process, and thus $\textbf{X}$ is a $p$-variate
Gaussian process with mean and variance (under $\mathbb{P}$) being bounded as long as $A$ has eigenvalues with negative real part.

Let us next analyse the $\mathbb{Q}$-dynamics of $\textbf{X}$ for the choice of $\theta$ given in \eqref{theta-carma}.
We find,
\begin{align*}
\ud\textbf{X}(t)&=\left(\boldsymbol{\xi}-\textbf{e}_p(\mu + \frac12\vartheta^2-r)-\textbf{e}_p
(\textbf{b}^{\top}\boldsymbol{\xi}-\textbf{c}^{\top}\boldsymbol{\theta})\right)\,\ud t \\
&\qquad+
\left(A\textbf{X}(t)-\textbf{e}_p(\textbf{b}^{\top}A-\textbf{c}^{\top}C)\textbf{X}(t)\right)\,\ud t
+\vartheta\textbf{e}_p\, \ud W_{\mathbb{Q}}(t)\,.
\end{align*}
Observe that $\textbf{e}_p(\textbf{b}^{\top}A\textbf{X}(t))=(\textbf{e}_p\textbf{b}^{\top})A\textbf{X}(t)$.
Hence, the second drift term in the $\mathbb{Q}$-dynamics of $\textbf{X}(t)$ above is therefore equal to
$(A-(\textbf{e}_p\textbf{b}^{\top})A+(\textbf{e}_p\textbf{c}^{\top})C)\textbf{X}(t)$. But
$$
\textbf{e}_p\textbf{b}^{\top}=\begin{bmatrix}
0 & 0 & 0 & \cdots & 0\\
0 & 0 & 0 & \cdots & 0\\
\vdots & \vdots & \vdots& \ddots & \vdots \\
0 & 0 & 0 & \cdots & 0\\
b_0 & b_1 & b_{2} & \cdots & b_{p-1}
\end{bmatrix}\,.
$$
where we recall $b_{p-1}=1$. It follows,
$$
(\textbf{e}_p\textbf{b}^{\top})A=\begin{bmatrix}
0 & 0 & 0 & \cdots & 0\\
0 & 0 & 0 & \cdots & 0\\
\vdots & \vdots & \vdots& \ddots & \vdots \\
0 & 0 & 0 & \cdots & 0\\
-\alpha_p & b_0-\alpha_{p-1} & b_1-\alpha_{p-2} & \cdots & b_{p-2}-\alpha_1
\end{bmatrix}\,.
$$
Similarly, for $\tilde{q}=p-1$, 
$$
(\textbf{e}_p\textbf{c}^{\top})C=\begin{bmatrix}
0 & 0 & 0 & \cdots & 0\\
0 & 0 & 0 & \cdots & 0\\
\vdots & \vdots & \vdots& \ddots & \vdots \\
0 & 0 & 0 & \cdots & 0\\
-c_{p-1}\beta_p & c_0-c_{p-1}\beta_{p-1} & c_1-c_{p-1}\beta_{p-2} & \cdots & c_{p-2}-c_{p-1}\beta_1
\end{bmatrix}\,.
$$
 We conclude that
\begin{align*}
&A-(\textbf{e}_p\textbf{b}^{\top})A+(\textbf{e}_p\textbf{c}^{\top})C\\
&\quad =\begin{bmatrix}
0 & 1 & 0 & \cdots & 0\\
0 & 0 & 1 & \cdots & 0\\
\vdots & \vdots & \vdots& \ddots & \vdots \\
0 & 0 & 0 & \cdots & 1\\
-c_{p-1}\beta_p & -b_0+c_0-c_{p-1}\beta_{p-1} & -b_1+c_1-c_{p-1}\beta_{p-2} & \cdots & -b_{p-2}+c_{p-2}-c_{p-1}\beta_1
\end{bmatrix}\,.
\end{align*}
Choosing $\textbf{c}=\textbf{b}$ yields that $A-(\textbf{e}_p\textbf{b}^{\top})A+(\textbf{e}_p\textbf{c}^{\top})C=C$. Therefore,
\begin{equation}\label{XunderQ}
\ud \textbf{X}(t)=(\widetilde{\boldsymbol{\xi}}+C\textbf{X}(t))\,\ud t+\vartheta\textbf{e}_p\, \ud W_\mathbb{Q}(t)\,,
\end{equation}
where
\begin{equation}
\widetilde{\boldsymbol{\xi}}=\boldsymbol{\xi}-\textbf{e}_p(\mu+\frac12\vartheta^2-r)+
\textbf{e}_p\textbf{b}^{\top}(\boldsymbol{\theta}-\boldsymbol{\xi})\,.
\end{equation}
Observe that $\widetilde{\boldsymbol{\xi}}$ only changes in the last coordinate compared to $\boldsymbol{\xi}$. In conclusion, we have that $\textbf{X}$ has $A$ as ``speed of mean-reversion'' under
$\mathbb{P}$, while $C$ under $\mathbb{Q}$. 

\begin{proposition}
	Let $S$ be as in \eqref{S-CARMA} with $\textbf{X}$ defined by the SDE \eqref{eq20} and let $q=p-1$. Assume that the eigenvalues of the matrix $A$ all have negative real part, the risk-free interest rate $r$ is constant and $\rho$ is defined by \eqref{convenience-spec-carma} with $\textbf{c}=\textbf{b}$. Then the forward price \eqref{forward} of the forward contract written on $S$ is given by
	\begin{equation}\label{FPCARMA}
	\begin{aligned}
	F(t,T)& = \exp\left(\mu T+\mathbf{b}^{\top}\exp(C(T-t))\mathbf X(t)+\mathbf{b}^{\top}\left(\exp(C(T-t))-I\right)C^{-1}\tilde{\boldsymbol{\xi}}
	\right) \\
	&\qquad\times\exp\left(\frac{\vartheta^2}{2}\int_0^{T-t}( \mathbf b^{\top}\exp(Cs)\mathbf e_p)^2\ud s\right)
		\end{aligned}
	\end{equation}
for $t\leq T$ and with $C$ the matrix \eqref{Cmatrix} and $I$ the $p\times p$ identity matrix.
\end{proposition}
\begin{proof}
Recalling that $S(T)$ equals $\exp (\mu T+ \textbf{b}^{\top}\textbf{X}(T))$ according to  \eqref{SofY}, with 
	\[
	\mathbf X(T)=\exp(C(T-t))\mathbf X(t)+\left(\exp(C(T-t))-I\right)C^{-1}\tilde{\boldsymbol{\xi}}+\vartheta\int\limits_t^T\exp(C(T-s))\mathbf{e}_p\,\ud W_\mathbb{Q}(s)
	\]
	as solution to \eqref{XunderQ}, the forward price can be expressed as
	\begin{align*}
		F(t,T) &=\mathbb{E}^{\mathbb{Q}}[S(T) \mid \mathcal{F}_t]\\
		& =\exp\left(\mu T+\mathbf{b}^{\top}\exp(C(T-t))\mathbf{X}(t)+\mathbf{b}^{\top}\left(\exp(C(T-t))-I\right)C^{-1}\tilde{\boldsymbol{\xi}}	\right) \\
			&\qquad\times \mathbb{E}^{\mathbb{Q}}\Big[\exp\Big(\vartheta\int\limits_t^T\mathbf{b}^{\top}\exp(C(T-s))\mathbf{e}_p\,\ud W_\mathbb{Q}(s)\Big)\Big].
	\end{align*}
	Evaluating the mean of a lognormal random variable combined with the isometry property provides the stated result.
\end{proof}

Note that solving the SDE \eqref{SDEforSunderQ} for $S$ leads to 
\[
F(t,T)=\mathbb{E}^{\mathbb{Q}}[S(T) \mid \mathcal{F}_t]= \mbox{e}^{r(T-t)}S(t)\mathbb{E}^{\mathbb{Q}}\Big[\exp \Big(\int\limits_t^T\rho(s) \mbox{d}s +\vartheta \big(W_\mathbb{Q}(T)-W_\mathbb{Q}(t)\big) \Big) \mid \mathcal{F}_t\Big],
\]
where $\rho(s)=\textbf{b}^{\top}(\boldsymbol{\theta}+C\textbf{X}(s))$ depends on the process $\textbf{X}$ that satisfies the SDE \eqref{XunderQ}. This approach however will not lead easily to the forward price \eqref{FPCARMA}.

\paragraph{Acknowledgments} The authors acknowledge the Centre of Advanced Study (CAS) at the Norwegian
Royal Academy of Science and Letters (Program SEFE) for providing occasions of research discussions
to start this paper.  Mich\`ele Vanmaele acknowledges the Research Foundation Flanders (FWO) and the
Special Research Fund (BOF) of the Ghent University for providing the possibility to go on sabbatical
leave to CAS. Fred Espen Benth is grateful for financial support from the research project ``FINEWSTOCH'',
funded by the Norwegian Research Council. We thank Fabian Andsem Harang for some discussions.




\bibliographystyle{plain}
\bibliography{biblioversieAP}
\end{document}